\documentclass[11pt,twoside,a4paper,notitlepage]{article}

\usepackage{ifthen}

\newboolean{conferenceversion}
\setboolean{conferenceversion}{false}

\ifthenelse{\boolean{conferenceversion}}
{}
{
  \usepackage[a4paper,margin=2.5cm]{geometry}
  \usepackage{times}
}

\ifthenelse{\boolean{conferenceversion}}
{}
{
  \usepackage[english]{babel}
}

\ifthenelse{\boolean{conferenceversion}}
{}
{
  \usepackage{fancyhdr}
}

\usepackage[sf,SF]{subfigure}

\usepackage{bussproofs}

\ifthenelse{\boolean{conferenceversion}}
{}
{
  
}

\makeatletter{}
\usepackage{ifthen}

\makeatletter{}

\usepackage{ifthen}

\newboolean{detectedSTOC}
\newboolean{detectedFOCS}
\newboolean{detectedElsevier}
\newboolean{detectedNOW}
\newboolean{detectedLMCS}
\newboolean{detectedPoster}
\newboolean{detectedSIAM}
\newboolean{detectedIEEE}
\newboolean{detectedLNCS}
\newboolean{detectedACM}
\newboolean{detectedToC}
\newboolean{detectedLIPIcs}
\newboolean{detectedArticle}
\newboolean{detectedReport}
\newboolean{detectedThesis}

\makeatletter

\@ifclassloaded{sig-alternate}
{\setboolean{detectedSTOC}{true}}
{\setboolean{detectedSTOC}{false}}

\@ifclassloaded{elsarticle}
{\setboolean{detectedElsevier}{true}}
{\setboolean{detectedElsevier}{false}}

\@ifclassloaded{now}
{\setboolean{detectedNOW}{true}}
{\setboolean{detectedNOW}{false}}

\@ifclassloaded{lmcs}
{\setboolean{detectedLMCS}{true}}
{\setboolean{detectedLMCS}{false}}

\@ifclassloaded{IEEEtran} {
  \setboolean{detectedIEEE}{true}
  \ifCLASSOPTIONconference {
    \setboolean{detectedFOCS}{true}
  }
  \else {
    \setboolean{detectedFOCS}{false}
  }
  \fi
}
{
  \setboolean{detectedFOCS}{false}
  \setboolean{detectedIEEE}{false}
}

\@ifclassloaded{siamltex1213} 
{\setboolean{detectedSIAM}{true}}
{\setboolean{detectedSIAM}{false}}

\@ifclassloaded{llncs}
{\setboolean{detectedLNCS}{true}}
{\setboolean{detectedLNCS}{false}}

\@ifclassloaded{acmsmall}
{\setboolean{detectedACM}{true}}
{\setboolean{detectedACM}{false}}

\@ifclassloaded{toc}
{\setboolean{detectedToC}{true}}
{\setboolean{detectedToC}{false}}

\@ifclassloaded{lipics}
{\setboolean{detectedLIPIcs}{true}}
{\setboolean{detectedLIPIcs}{false}}

\@ifclassloaded{sciposter}
{\setboolean{detectedPoster}{true}}
{\setboolean{detectedPoster}{false}}

\@ifclassloaded{article}
{\setboolean{detectedArticle}{true}}
{\setboolean{detectedArticle}{false}}

\makeatother

\ifthenelse{\not \isundefined{\examen} 
  \and \not \isundefined{\disputationsdatum} 
  \and \not \isundefined{\disputationslokal}}   
  {\setboolean{detectedThesis}{true}}
  {\setboolean{detectedThesis}{false}}
           
\ifthenelse{\boolean{detectedArticle} \or \boolean{detectedThesis}
  \or \boolean{detectedSTOC} \or \boolean{detectedFOCS}
  \or \boolean{detectedSIAM} \or \boolean{detectedIEEE}
  \or \boolean{detectedPoster}}
{\setboolean{detectedReport}{false}}
{\setboolean{detectedReport}{true}}

\ifthenelse{\boolean{detectedLIPIcs}}
{}
{\usepackage[utf8]{inputenc}}

\usepackage{amsmath}
\usepackage{amssymb}
\usepackage{amsfonts}

\usepackage{mathtools}

\ifthenelse
{\boolean{detectedElsevier}}
{}
{\usepackage{varioref}}

\usepackage{xspace}

\ifthenelse    
{\boolean{detectedSTOC}   \or \boolean{detectedSIAM} \or 
  \boolean{detectedLMCS}  \or \boolean{detectedNOW}  \or 
  \boolean{detectedLNCS}  \or \boolean{detectedACM}  \or
  \boolean{detectedToC}  \or \boolean{detectedLIPIcs}}
{}
{
  \usepackage{amsthm}
  \usepackage{amsthm-boldfont}
}

\ifthenelse        
{\boolean{detectedElsevier}  \or \boolean{detectedFOCS}   \or 
  \boolean{detectedSTOC}     \or \boolean{detectedPoster} \or
  \boolean{detectedSIAM}     \or \boolean{detectedLMCS}   \or
  \boolean{detectedIEEE}     \or \boolean{detectedNOW}    \or
  \boolean{detectedToC}      \or \boolean{detectedThesis} \or
  \boolean{detectedLNCS}     \or \boolean{detectedACM}    \or 
  \boolean{detectedLIPIcs}}
{}
{\usepackage{nada-typography}}

\makeatletter{}

\usepackage{ifthen}
\usepackage{xspace}
\ifthenelse
{\boolean{detectedElsevier}}
{}
{\usepackage{varioref}}

\DeclareMathAlphabet{\mathsfsl}{OT1}{cmss}{m}{sl}

\newcommand{\formatfunctiontoset}[1]{\mathit{#1}}

\newcommand{\introduceterm}[1]{{\emph{#1}}}

\newcommand{\eqperiod}{\enspace .}
\newcommand{\eqcomma}{\enspace ,}

\newcommand{\wrt}{with respect to\xspace}

\ifthenelse{\boolean{detectedToC}}{}
  {\newcommand{\eg}{for instance\xspace}     }
     
\ifthenelse{\boolean{detectedToC}}{}{\newcommand{\ie}{i.e.,\ }}

\ifthenelse{\boolean{detectedLIPIcs}}
{}
{}

\newcommand{\etal}{et al.\@\xspace}

\ifthenelse{\boolean{detectedIEEE}}{}{}

\newcommand{\Bigoh}[1]{\mathrm{O} \bigl( #1 \bigr)}

\newcommand{\Bigtheta}[1]{\Theta \bigl( #1 \bigr)}

\newcommand{\BIGOMEGA}[1]{\Omega \left( #1 \right)}
\newcommand{\Bigomega}[1]{\Omega \bigl( #1 \bigr)}
\newcommand{\bigomega}[1]{\Omega ( #1 )}

\newcommand{\complclassformat}[1]        {\textrm{\upshape{\textsf{#1}}}\xspace}

\newcommand{\cocomplclass}[1]        {\textrm{\upshape{\textsf{co#1}}}\xspace}

\newcommand{\NP}{\complclassformat{NP}}

\newcommand{\coNP}{\cocomplclass{NP}}

\newcommand{\refsec}[1]{Section~\ref{#1}}

\newcommand{\refth}[1]{Theorem~\ref{#1}}

\newcommand{\reflem}[1]{Lemma~\ref{#1}}

\newcommand{\refcor}[1]{Corollary~\ref{#1}}

\newcommand{\refdef}[1]{Definition~\ref{#1}}

\newcommand{\refobs}[1]{Observation~\ref{#1}}

\newcommand{\reffact}[1]{Fact~\ref{#1}}

\newcommand{\Refth}[1]{Theorem~\ref{#1}}

\newcommand{\Refdef}[1]{Definition~\ref{#1}}

\ifthenelse
{\isundefined{\refeq}}
{\newcommand{\refeq}[1]{\eqref{#1}}}
{\renewcommand{\refeq}[1]{\eqref{#1}}}

\newcommand{\floor}[1]{\lfloor #1 \rfloor}

\ifthenelse{\boolean{detectedToC}}{}
{

  \newcommand{\Nplus}     {\mathbb{N}^{+}}

}

\newcommand{\maxofexpr}[2][]{\max_{#1} \{ #2 \}}

\newcommand{\fieldstd}{\mathbb{F}}
\newcommand{\F}{\mathbb{F}}

\ifthenelse{\boolean{detectedLMCS}}
{}
{}

\newcommand{\twincommandJN}[6]    {#1#2#3\vphantom{#2#5}\mspace{-2.25mu}#4.#5#6}

\newcommand{\funcdescr}[3]{\ensuremath{ #1 : #2 \to #3}}

\newcommand{\domainof}[1]{\ensuremath{\mathrm{Dom} ( #1 )}}

 \newcommand{\vertices}[1]{V( #1 )}

\newcommand{\descname}{\formatfunctiontoset{desc}}

\ifthenelse{\isundefined{\descnode}}
{\newcommand{\descnode}[2][]{\descname_{#1}(#2)}}
{\renewcommand{\descnode}[2][]{\descname_{#1}(#2)}}

\newcommand{\setsizesmall}[1]{\ensuremath{\lvert#1\rvert}}

\newcommand{\set}[1]{\{ #1 \}}
\newcommand{\Set}[1]{\bigl\{ #1 \bigr\}}

\newcommand{\setdescr}[3][\mid]{\set{ #2 #1 #3 }}
\newcommand{\Setdescr}[3][|]     {\ifthenelse{\equal{#1}{;}}     {\Set{ #2 \,;\, #3 }}
     {\ifthenelse{\equal{#1}{:}}     {\Set{ #2 \,:\, #3 }}
     {\twincommandJN{\bigl\{}{#2}{\bigl#1}{\bigr}{\,#3}{\bigr\}}}}}

\newcommand{\Setsize}[1]{\bigl\lvert#1\bigr\rvert}
\newcommand{\setsize}[1]{\lvert#1\rvert}

\newcommand{\intersection}{\cap}

\newcommand{\union}{\cup}
\newcommand{\Union}{\bigcup}

\newcommand{\disjointunion}{\overset{.}{\cup}}

\newcommand{\Lor}{\bigvee}
\newcommand{\Land}{\bigwedge}

\newcommand{\olnot}[1]{\overline{#1}}

\newcommand{\clwidth}{k}

\newcommand{\inductionformat}[1]{\textit{#1}}

\newcommand{\BASE}[1][]
        {\inductionformat
                {                        \ifthenelse{\equal{#1}{}}                                {Base case: }                                {Base case (#1):}                }        }

\ifthenelse{\isundefined{\DONOTLOADHYPERREFPACKAGE}
  \and \not \boolean{detectedSTOC}     \and \not \boolean{detectedFOCS}
  \and \not \boolean{detectedElsevier} \and \not \boolean{detectedPoster} 
  \and \not \boolean{detectedSIAM}     \and \not \boolean{detectedACM}
  \and \not \boolean{detectedIEEE}     \and \not \boolean{detectedNOW}
  \and \not \boolean{detectedToC}      \and \not \boolean{detectedThesis}
  \and \not \boolean{detectedLNCS}     \and \not \boolean{detectedLIPIcs}}
{\usepackage[colorlinks=true,citecolor=blue]{hyperref}}
{}

\ifthenelse{\boolean{detectedSTOC}}                  
  {\makeatletter{}

\newtheorem{theorem}{Theorem}
\newtheorem{lemma}[theorem]{Lemma}
\newtheorem{proposition}[theorem]{Proposition}
\newtheorem{corollary}[theorem]{Corollary}
\newtheorem{observation}[theorem]{Observation}

\newtheorem{definition}[theorem]{Definition}
\newtheorem{claim}[theorem]{Claim}

\newtheorem{conjecture}{Conjecture}
\newtheorem{openproblem}[conjecture]{Open Problem}

\newdef{remark}{Remark}
\newdef{example}{Example}

\newcounter{unnumber}

 }
  {\ifthenelse{\boolean{detectedSIAM}}
    {\makeatletter{}

\newtheorem{observation}[theorem]{Observation}

\newtheorem{conjecture}{Conjecture}
\newtheorem{openquestion}{Open Question}

\newtheorem{remarkinner}[theorem]{Remark}
\newtheorem{exampleinner}[theorem]{Example}

\newcommand{\exampleendmarker}{\qquad$\Diamond$}
\newcommand{\remarkendmarker}{\qquad$\Diamond$}

\newenvironment{example}                        
    {\begin{exampleinner} \rm}
    {\exampleendmarker\end{exampleinner}}

\newenvironment{remark}                        
    {\begin{remarkinner} \rm}
    {\remarkendmarker\end{remarkinner}}

\newcounter{unnumber}

 }
    {\ifthenelse{\boolean{detectedNOW}}
      {\makeatletter{}

\newtheorem{standardlocalcounter}{Dummy}[chapter]
\newtheorem{standardglobalcounter}{Dummy}

\newtheorem{theorem}[standardlocalcounter]{Theorem}
\newtheorem{lemma}[standardlocalcounter]{Lemma}
\newtheorem{proposition}[standardlocalcounter]{Proposition}
\newtheorem{corollary}[standardlocalcounter]{Corollary}
\newtheorem{observation}[standardlocalcounter]{Observation}
\newtheorem{fact}[standardlocalcounter]{Fact}

\newtheorem{conjecturelocalcounter}[standardlocalcounter]{Conjecture}
\newtheorem{conjectureglobalcounter}[standardglobalcounter]{Conjecture}
\newtheorem{conjecture}[standardglobalcounter]{Conjecture}
\newtheorem{openquestion}[standardglobalcounter]{Open Question}
\newtheorem{openproblem}[standardglobalcounter]{Open Problem}
\newtheorem{problem}{Problem}

\newtheorem{property}[standardlocalcounter]{Property}
\newtheorem{definition}[standardlocalcounter]{Definition}
\newtheorem{claim}[standardlocalcounter]{Claim}

\newtheorem{algorithm}[standardlocalcounter]{Algorithm}

\newtheorem{remark}[standardlocalcounter]{Remark}
\newtheorem{example}[standardlocalcounter]{Example}

\renewenvironment{proof}[1][Proof]{\par\trivlist
   \item[\hskip \labelsep{\itshape {#1}.}]\prooffont}
   {\hspace*{0pt plus1fill}\fboxsep2.5pt\fboxrule.5pt\raise3pt\hbox{\fbox{}}\endtrivlist}
 }
      {\ifthenelse{\boolean{detectedLMCS}}
        {\makeatletter{}

\theoremstyle{plain}    

\newtheorem{theorem}[thm]{Theorem}
\newtheorem{lemma}[thm]{Lemma}
\newtheorem{proposition}[thm]{Proposition}
\newtheorem{corollary}[thm]{Corollary}
\newtheorem{observation}[thm]{Observation}

\newtheorem{conjecture}[thm]{Conjecture}
\newtheorem{problem}[thm]{Problem}

\newtheorem{openquestion}{Open Question}
\newtheorem{openproblem}{Open Problem}

\theoremstyle{definition}

\newtheorem{property}[thm]{Property}
\newtheorem{definition}[thm]{Definition}
\newtheorem{claim}[thm]{Claim}
\newtheorem{remark}[thm]{Remark}
\newtheorem{example}[thm]{Example}

 }
        {\ifthenelse{\boolean{detectedIEEE}}
          {\makeatletter{}

\newtheorem{standardlocalcounter}{Dummy}[section]
\newtheorem{standardglobalcounter}{Dummy}

\theoremstyle{plain}    

\newtheorem{theorem}[standardglobalcounter]{Theorem}
\newtheorem{lemma}[standardglobalcounter]{Lemma}
\newtheorem{proposition}[standardglobalcounter]{Proposition}
\newtheorem{corollary}[standardglobalcounter]{Corollary}
\newtheorem{observation}[standardglobalcounter]{Observation}
\newtheorem{fact}[standardglobalcounter]{Fact}

\newtheorem{conjecture}[standardglobalcounter]{Conjecture}
\newtheorem{openquestion}{Open Question}
\newtheorem{openproblem}{Open Problem}
\newtheorem{problem}{Problem}

\theoremstyle{definition}

\newtheorem{property}[standardglobalcounter]{Property}
\newtheorem{definition}[standardglobalcounter]{Definition}
\newtheorem{claim}[standardglobalcounter]{Claim}

\theoremstyle{remark}
\newtheorem{remark}[standardglobalcounter]{Remark}
\newtheorem{example}[standardglobalcounter]{Example}

\newtheoremstyle{meta}  {3pt}  {3pt}  {\scshape \small }  {}  {\scshape \small }  {:}  { }          {}
\theoremstyle{meta}
\newtheorem{meta}{Meta comment}

\newtheoremstyle{questions}  {3pt}  {3pt}  {\sffamily \slshape}  {}  {\bfseries \sffamily \slshape}  {:}  { }          {}
\theoremstyle{questions}
\newtheorem{questions}{Open questions}

 }
          {\ifthenelse{\boolean{detectedLNCS}}
            {\makeatletter{}

\spnewtheorem*{proofsketch}{Proof sketch}{\itshape}{\rmfamily}

\spnewtheorem{observation}{Observation}{\bfseries}{\itshape}
\spnewtheorem{fact}{Fact}{\bfseries}{\itshape}

 }
            {\ifthenelse{\boolean{detectedACM}}
              {\makeatletter{}

\newtheorem{conjecture}[theorem]{Conjecture}
\newtheorem{observation}[theorem]{Observation}
\newtheorem{claim}[theorem]{Claim}
\newtheorem{openquestion}{Open Question}

\newcounter{unnumber}

 }
              {\ifthenelse{\boolean{detectedToC}}
                {\makeatletter{}

\theoremstyle{plain}
\newtheorem{observation}[theorem]{Observation}
\newtheorem{openproblem}[theorem]{Open Problem}

\theoremstyle{definition}
\newtheorem{property}[theorem]{Property}

\renewcommand{\refth}[1]{\expref{Theorem}{#1}}
\renewcommand{\reflem}[1]{\expref{Lemma}{#1}}

\renewcommand{\refcor}[1]{\expref{Corollary}{#1}}
\renewcommand{\refdef}[1]{\expref{Definition}{#1}}

\renewcommand{\refobs}[1]{\expref{Observation}{#1}}

\renewcommand{\Refth}[1]{\expref{Theorem}{#1}}

\renewcommand{\Refdef}[1]{\expref{Definition}{#1}}

\renewcommand{\refsec}[1]{\expref{Section}{#1}}

 }
                {\ifthenelse{\boolean{detectedLIPIcs}}
                  {\makeatletter{}

\theoremstyle{plain}    

\newtheorem{fact}[theorem]{Fact}
\newtheorem{proposition}[theorem]{Proposition}
\newtheorem{observation}[theorem]{Observation}

 }
                  {\ifthenelse{\boolean{detectedArticle} 
                      \or \boolean{detectedElsevier}}
                    {\makeatletter{}

\newtheorem{standardlocalcounter}{Dummy}[section]
\newtheorem{standardglobalcounter}{Dummy}

\theoremstyle{plain}    

\newtheorem{theorem}[standardlocalcounter]{Theorem}
\newtheorem{lemma}[standardlocalcounter]{Lemma}
\newtheorem{proposition}[standardlocalcounter]{Proposition}
\newtheorem{corollary}[standardlocalcounter]{Corollary}
\newtheorem{observation}[standardlocalcounter]{Observation}
\newtheorem{fact}[standardlocalcounter]{Fact}

\newtheorem{conjecturelocalcounter}[standardlocalcounter]{Conjecture}
\newtheorem{conjectureglobalcounter}[standardglobalcounter]{Conjecture}
\newtheorem{conjecture}[standardglobalcounter]{Conjecture}
\newtheorem{openquestion}[standardglobalcounter]{Open Question}
\newtheorem{openproblem}[standardglobalcounter]{Open Problem}
\newtheorem{problem}[standardglobalcounter]{Problem}

\theoremstyle{definition}

\newtheorem{property}[standardlocalcounter]{Property}
\newtheorem{definition}[standardlocalcounter]{Definition}
\newtheorem{claim}[standardlocalcounter]{Claim}

\theoremstyle{remark}
\newtheorem{remark}[standardlocalcounter]{Remark}
\newtheorem{example}[standardlocalcounter]{Example}

\newtheoremstyle{meta}  {3pt}  {3pt}  {\scshape \small }  {}  {\scshape \small }  {:}  { }          {}
\theoremstyle{meta}
\newtheorem{meta}{Meta comment}

\newtheoremstyle{questions}  {3pt}  {3pt}  {\sffamily \slshape}  {}  {\bfseries \sffamily \slshape}  {:}  { }          {}
\theoremstyle{questions}
\newtheorem{questions}{Open questions}

 }
                    {\makeatletter{}

\newtheorem{standardlocalcounter}{Dummy}[chapter]
\newtheorem{standardglobalcounter}{Dummy}

\theoremstyle{plain}    

\newtheorem{theorem}[standardlocalcounter]{Theorem}
\newtheorem{lemma}[standardlocalcounter]{Lemma}
\newtheorem{proposition}[standardlocalcounter]{Proposition}
\newtheorem{corollary}[standardlocalcounter]{Corollary}
\newtheorem{observation}[standardlocalcounter]{Observation}

\theoremstyle{definition}

\newtheorem{definition}[standardlocalcounter]{Definition}

\theoremstyle{remark}

\newtheorem{example}[standardlocalcounter]{Example}

\newtheoremstyle{meta}  {3pt}  {3pt}  {\scshape \small }  {}  {\scshape \small }  {:}  { }          {}
\theoremstyle{meta}

\newtheoremstyle{questions}  {3pt}  {3pt}  {\sffamily \slshape}  {}  {\bfseries \sffamily \slshape}  {:}  { }          {}
\theoremstyle{questions}

 }}}}}}}}}}

\ifthenelse{\boolean{detectedThesis}}
{}
{\usepackage{ifpdf}}

\ifthenelse{\boolean{detectedToC}}
{}
{\usepackage{graphicx}}

\ifpdf         
\fi

\ifthenelse
{\boolean{detectedSTOC}  \or \boolean{detectedThesis} \or 
 \boolean{detectedLNCS}  \or \boolean{detectedToC}    \or 
 \boolean{detectedLIPIcs}}    
{}
{
  \setcounter{secnumdepth}{3}
  \setcounter{tocdepth}{3}
}

\newcount\hours
\newcount\minutes
\def\SetTime{\hours=\time
\global\divide\hours by 60
\minutes=\hours
\multiply\minutes by 60
\advance\minutes by-\time
\global\multiply\minutes by-1 }
\SetTime
\def\now{\number\hours:\ifnum\minutes<10 0\fi\number\minutes}

\makeatletter{}

\newcommand{\formuladots}{\cdots}

\newcommand{\boundary}[1]{\ensuremath{\partial #1}}

\newcommand{\proofstd}{\ensuremath{\pi}}

\DeclareFontFamily{OT1}{pzc}{}
\DeclareFontShape{OT1}{pzc}{m}{it}{<-> s * [1.200] pzcmi7t}{}
\DeclareMathAlphabet{\mathpzc}{OT1}{pzc}{m}{it}

\newcommand{\proofsystemsubindexformat}[1]{\scriptscriptstyle{\mathcal{#1}}}

\newcommand{\pcnot}{\proofsystemsubindexformat{PC}}
\newcommand{\pcrnot}{\proofsystemsubindexformat{PCR}}

\newcommand{\derivof}[4][\derives]
        {{\ensuremath{{#2} : {#3} \, {#1}\, {#4}}}}

\newcommand{\refof}[2]{\derivof{#1}{#2}{\falsenum}}
\renewcommand{\refof}[2]{\derivof{#1}{#2}{\emptycl}}
\renewcommand{\refof}[2]{\derivof{#1}{#2}{\bot}}

\newcommand{\fstd}{{\ensuremath{F}}}

\newcommand{\formf}{\ensuremath{F}}

\newcommand{\varx}{\ensuremath{x}}

\newcommand{\lita}{\ensuremath{a}}

\newcommand{\clc}{\ensuremath{C}}

\newcommand{\termt}{\ensuremath{T}}

\newcommand{\setsofvarsorlitsmall}[2]
        {\mathit{#1}({#2})}

\newcommand{\vars}[1]{\setsofvarsorlitsmall{Vars}{#1}}

\newcommand{\restr}{\rho}
\newcommand{\rstd}{\restr}

\newcommand{\restrict}[2]{#1\!\!\upharpoonright_{#2}}

\newcommand{\derivabbrevsmall}[2]{( #1 \vdash #2 )}

\newcommand{\refutabbrevsmall}[1]{\derivabbrevsmall{#1}{\falsenum}}

\renewcommand{\refutabbrevsmall}[1]{\derivabbrevsmall{#1}{\!\emptycl}}

\renewcommand{\refutabbrevsmall}[1]{\derivabbrevsmall{#1}{\!\bot}}

\newcommand{\genericformsmall}[2]{\mathit{#1}( #2 )}

\newcommand{\genericrefsmall}[3]    {{\mathit{#1}}_{#2}\refutabbrevsmall{#3}}

\newcommand{\sizeofarg}[1]{\genericformsmall{S}{#1}}

\newcommand{\sizeref}[2][]{\genericrefsmall{S}{#1}{#2}}

\newcommand{\lengthofarg}[1]{\genericformsmall{L}{#1}}

\newcommand{\lengthref}[2][]{\genericrefsmall{L}{#1}{#2}}

\newcommand{\widthofsmall}[2][]{\genericformsmall{W_{#1}}{#2}}

\newcommand{\widthofarg}[2][]{\genericformsmall{W_{#1}}{#2}}

\newcommand{\mdegreeof}[2][]{\genericformsmall{Deg_{#1}}{#2}}

\newcommand{\mdegreeref}[2][]{\genericrefsmall{Deg}{#1}{#2}}

\newcommand{\formulaformat}[1]{\ensuremath{\mathit{#1}}}
\renewcommand{\formulaformat}[1]{\mathit{#1}}

\newcommand{\stoptime}{\tau}

\newcommand{\GraphPHPnot}[1][G]{\formulaformat{PHP}_{#1}}
\newcommand{\GraphFPHPnot}[1][G]{\formulaformat{FPHP}_{#1}}
\newcommand{\graphphpnot}[1][G]{\formulaformat{PHP}_{#1}}
\newcommand{\graphfphpnot}[1][G]{\formulaformat{FPHP}_{#1}}
\newcommand{\GraphOntoPHPnot}[1][G]    {\text{$\formulaformat{Onto}$-$\formulaformat{PHP}$}_{#1}}
\newcommand{\GraphOntoFPHPnot}[1][G]
    {\text{$\formulaformat{Onto}$-$\formulaformat{FPHP}$}_{#1}}

\newcommand{\fphpnot}[2]{\formulaformat{FPHP}^{#1}_{#2}}

\makeatletter{}

\usepackage{stmaryrd}

\makeatletter{}
\newcommand{\paramn}{n}

\newcommand{\assignmenta}{\rho}

\newcommand{\disunion}{\mathbin{\dot{\cup}}}

\newcommand{\graphg}{G}

\newcommand{\vertexsetu}{U}
\newcommand{\vertexsetv}{V}

\newcommand{\vertexu}{u}
\newcommand{\vertexv}{v}
\newcommand{\vertexw}{w}

\newcommand{\edgesete}{E}

\newcommand{\edgee}{e}

\newcommand{\neighbor}[1]{N(#1)}
\newcommand{\Neighbor}[1]{N\bigl(#1\bigr)}
\newcommand{\NEIGHBOR}[1]{N\left(#1\right)}

\newcommand{\expsize}{s}
\newcommand{\expszfct}{\epsilon}
\renewcommand{\expszfct}{\gamma}
\newcommand{\expfactor}{\delta}
\newcommand{\expslack}{\xi}

\newcommand{\originalcnf}{\mathcal{F}}
\newcommand{\fpartition}{\mathcal{F}}
\renewcommand{\fpartition}{\mathcal{U}}
\newcommand{\vpartition}{\mathcal{V}}

\newcommand{\fixedconstraints}{E}

\newcommand{\fvstructurenot}{(\fpartition, \vpartition)_{\fixedconstraints}}
\newcommand{\fvstructure}    {$\fvstructurenot$\nobreakdash-graph\xspace}
\newcommand{\afvstructure}{an \fvstructure}
\newcommand{\Afvstructure}{An \fvstructure}
\renewcommand{\afvstructure}{a \fvstructure}
\renewcommand{\Afvstructure}{A \fvstructure}

\newcommand{\fvstructurename}{\mathcal{G}}

\newcommand{\variablesv}{V}

\newcommand{\sboundary}[1]{\partial_{\fixedconstraints}(#1)}
\newcommand{\Sboundary}[1]{\partial_{\fixedconstraints}\bigl(#1\bigr)}
\newcommand{\SBOUNDARY}[1]{\partial_{\fixedconstraints}\left(#1\right)}

\newcommand{\rboundaryexpander}[4]    {$({#1}, {#2}, {#3}, {#4})$-\rsatrespectful boundary expander\xspace}
\newcommand{\rboundaryexpanderstd}    {\rboundaryexpander{\expsize}{\expfactor}{\expslack}{\fixedconstraints}}

\newcommand{\rboundaryexpandershort}[4]    {$({#1}, {#2}, {#3}, {#4})$-expander\xspace}
\newcommand{\rboundaryexpandershortstd}    {\rboundaryexpandershort      {\expsize}{\expfactor}{\expslack}{\fixedconstraints}}

\newcommand{\spanof}[2]{\mathcal{I}_{#1}(#2)}

\newcommand{\spanofc}[1][\fpartition']{\spanof{\fixedconstraints}{#1}}

\newcommand{\rsatrespects}{respects\xspace}
\newcommand{\rsatrespect}{respect\xspace}
\newcommand{\rsatrespectful}{respectful\xspace}
\newcommand{\Rsatrespectful}{Respectful\xspace}
\newcommand{\rsatrespectfulness}{respectfulness\xspace}
\newcommand{\rsatdisrespectful}{disrespectful\xspace}

\newcommand{\radjective}[1][$\fixedconstraints$]           {{#1}\nobreakdash-respectful\xspace}
\newcommand{\rnegadjective}[1][$\fixedconstraints$]           {{#1}\nobreakdash-disrespectful\xspace}
\newcommand{\radverb}[1][$\fixedconstraints$]           {{#1}\nobreakdash-respectfully\xspace}

\newcommand{\rsatisfiable}[1][$\fixedconstraints$]    {\radverb[#1] satisfiable\xspace}
\newcommand{\rsatisfies}[1][$\fixedconstraints$]    {\radverb[#1] satisfies\xspace}
\newcommand{\rsatisfy}[1][$\fixedconstraints$]    {\radverb[#1] satisfy\xspace}

\newcommand{\rneighbour}[1][$\fixedconstraints$]    {\radjective[#1] neighbour\xspace}
\newcommand{\anrneighbour}[1][$\fixedconstraints$]    {an \radjective[#1] neighbour\xspace}
\newcommand{\rneighbours}[1][$\fixedconstraints$]    {\radjective[#1] neighbours\xspace}

\newcommand{\rnegneighbours}[1][$\fixedconstraints$]    {\rnegadjective[#1] neighbours\xspace}

\newcommand{\rboundary}[1][$\fixedconstraints$]    {\radjective[#1] boundary\xspace}

\newcommand{\voverlapterm}{overlap\xspace}
\newcommand{\voverlapnot}[1]{\mathit{ol}({#1})}
\newcommand{\xoverlapnot}[2]{\mathit{ol}({#1}, {#2})}

\renewcommand{\boundary}[1]{\partial(#1)}

\newcommand{\redophead}{R}
\newcommand{\redop}[2][]{\redophead_{#1}(#2)}

\newcommand{\fvredophead}[1][\fvstructurename]{R_{#1}}
\newcommand{\fvredop}[2][\fvstructurename]{\fvredophead[#1](#2)}

\newcommand{\mleq}{\preccurlyeq}
\newcommand{\mlt}{\prec}

\newcommand{\polyax}{C}
\newcommand{\monomialm}{m}
\renewcommand{\termt}{t}

\newcommand{\ideali}{I}
\newcommand{\idealgen}[1]{\langle {#1} \rangle}
\newcommand{\leadingterm}[1]{\mathit{LT}({#1})}

\newcommand{\polyp}{P}
\newcommand{\polyq}{Q}
\newcommand{\polyr}{R}

\newcommand{\idpolyq}{Q}
\newcommand{\redpolyr}{R}

\newcommand{\litpos}{L^{+}}
\newcommand{\litneg}{L^{-}}

\newcommand{\setdegd}{d}
\newcommand{\degd}{D}
\newcommand{\graphdeg}{d}

\newcommand{\svcontainedadj}{contained\xspace}

\newcommand{\svcontained}[2]    {$({#1}, {#2})$\nobreakdash-\svcontainedadj}

\newcommand{\svcontainedstd}{\svcontained{\expsize}{\vpartition'}}
\newcommand{\shalfvcontainedstd}{\svcontained{\expsize/2}{\vpartition'}}

\newcommand{\sntcontained}{\svcontained{\expsize}{\neighbor{\termt}}}
\newcommand{\svcontainedneighbourt}{\sntcontained}

\newcommand{\ssupport}{$\expsize$\nobreakdash-support\xspace}

\newcommand{\support}[2][\expsize]{\mathit{Sup}_{#1}(#2)}

\newcommand{\gop}[1][\graphg]{\formulaformat{GOP}(#1)}

\newcommand{\lop}[1][n]{\formulaformat{LOP}_{#1}}

\newcommand{\subcardnot}[1]{\formulaformat{SC}(#1)}

\makeatletter{}

\newcommand{\TheauthorJN}{The second author\xspace}

\ifthenelse{\boolean{conferenceversion}}
{

}
{

}

\ifthenelse{\boolean{conferenceversion}}
{}
{
\newtheoremstyle{metacommenttheoremstyle}    {3pt}    {3pt}    {\sffamily \itshape \scriptsize
          }    {}    {\bfseries \scshape \footnotesize }    {:}    { }        {}
\theoremstyle{metacommenttheoremstyle}
}
\newtheorem{jncommentcontainer}{Jakob's comment}
\newtheorem{mlcommentcontainer}{Massimo's comment}
\newtheorem{mmcommentcontainer}{Mladen's comment}
\newtheorem{mvcommentcontainer}{Marc's comment}

\ifthenelse{\isundefined{\DONOTINSERTCOMMENTS}}
{   \usepackage[usenames,dvipsnames]{xcolor}

  \newcommand{\jncomment}[1]  {\begin{jncommentcontainer} \textcolor{blue}{#1} \end{jncommentcontainer}}

  \newcommand{\mlcomment}[1]  {\begin{mlcommentcontainer} \textcolor{OliveGreen}{#1} \end{mlcommentcontainer}}

  \newcommand{\mmcomment}[1]  {\begin{mmcommentcontainer} \textcolor{magenta}{#1} \end{mmcommentcontainer}}

  \newcommand{\mvcomment}[1]  {\begin{mvcommentcontainer} \textcolor{orange}{#1} \end{mvcommentcontainer}}

}
{   \newcommand{\jncomment}[1]{}
  \newcommand{\mlcomment}[1]{}
  \newcommand{\mmcomment}[1]{}
  \newcommand{\mvcomment}[1]{}
}

\newcommand{\jncommentdelete}[1]{}
\newcommand{\mmcommentdelete}[1]{}

\ifthenelse{\boolean{conferenceversion}}
{}
{
  \numberwithin{equation}{section}
}

\begin{document}

\title{A Generalized Method for Proving  \\  
  Polynomial Calculus Degree Lower Bounds  \thanks{This is the full-length version of the paper with the same
    title to appear in 
    \emph{Proceedings of the 30th Annual Computational Complexity
      Conference ({CCC}~'15)}.}}

\author{  Mladen Mik\v{s}a \\
  KTH Royal Institute of Technology
  \and
  Jakob Nordström \\
  KTH Royal Institute of Technology}

\date{\today}

\maketitle

\ifthenelse{\boolean{conferenceversion}}
{}
{

    \thispagestyle{empty}

    \pagestyle{fancy}
    \fancyhead{}
  \fancyfoot{}
    \renewcommand{\headrulewidth}{0pt}
  \renewcommand{\footrulewidth}{0pt}
  
                    \fancyhead[CE]{\slshape 
    A GENERALIZED METHOD FOR PROVING PC DEGREE LOWER BOUNDS
}

  \fancyhead[CO]{\slshape \nouppercase{\leftmark}}
  \fancyfoot[C]{\thepage}
  
            \setlength{\headheight}{13.6pt}
}

\thispagestyle{empty}

\makeatletter{}\begin{abstract}
  We study the problem of obtaining lower bounds for polynomial
  calculus (PC) and polynomial calculus resolution (PCR) on proof
  degree, and hence by [Impagliazzo \etal~'99] also on proof size.
  [Alekhnovich and Razborov~'03] established that if the clause-variable
  incidence graph of a CNF formula~$F$ is a good enough expander, then
  proving that~$F$ is unsatisfiable requires high PC/PCR degree. We
  further develop the techniques in [AR03] to show that if one can
  ``cluster'' clauses and variables in a way that ``respects the
  structure'' of the formula in a certain sense, then it is sufficient
  that the incidence graph of this clustered version is an
  expander. As a corollary of this, we prove that the functional
  pigeonhole principle (FPHP) formulas require high PC/PCR degree when
  restricted to constant-degree expander graphs. This answers an open
  question in [Razborov~'02], and also implies that the standard CNF
  encoding of the FPHP formulas require exponential proof size in
  polynomial calculus resolution. Thus, while Onto-FPHP formulas are
  easy for polynomial calculus, as shown in [Riis~'93], both FPHP and
  Onto-PHP formulas are hard even when restricted to bounded-degree
  expanders. 
\end{abstract}

\makeatletter{}\section{Introduction}
\label{sec:intro}

In one sentence, proof complexity studies how hard it is to certify
the unsatifiability of formulas in conjunctive normal form (CNF).  In
its most general form, this is the question of whether \coNP can be
separated from \NP or not, and as such it still appears almost
completely out of reach. However, if one instead focuses on concrete
proof systems, which can be thought of as restricted models of
(nondeterministic) computation, then fruitful study is possible.

\subsection{Resolution and Polynomial Calculus}

Perhaps the most well-studied proof system in proof complexity is
\introduceterm{resolution}~\cite{Blake37Thesis}, in which one derives
new disjunctive clauses from 
a CNF formula 
until an explicit
contradiction is reached, and for which numerous exponential lower
bounds on proof size have been shown 
\ifthenelse{\boolean{conferenceversion}}
{(starting with
  \cite{CS88ManyHard,Haken85Intractability,Urquhart87HardExamples}).}
{(starting with
  \cite{Haken85Intractability,Urquhart87HardExamples,CS88ManyHard}).}
Many
of these lower bounds can be established by instead studying the
\introduceterm{width} of proofs, \ie the size of a largest clause
appearing in the proofs, and arguing that any resolution proof for a certain
formula must contain a large clause. It then follows from a result by
Ben-Sasson and Wigderson~\cite{BW01ShortProofs} that any resolution
proof must also consist of very many clauses. Research
since~\cite{BW01ShortProofs} has led to a well-developed machinery for
showing width lower bounds, and hence also size lower bounds.

The focus of the current paper is the slightly more general proof
system \introduceterm{polynomial calculus resolution (PCR)}. 
This proof system was introduced by Clegg \etal~\cite{CEI96Groebner}
in a slightly weaker form that is usually referred to as
\introduceterm{polynomial calculus (PC)} and was later extended by 
Alekhnovich \etal~\cite{ABRW02SpaceComplexity}.
In PC and PCR clauses are
translated to multilinear polynomials over some (fixed) field~$\F$,
and a CNF formula~$\fstd$ is shown to be unsatisfiable by proving that
the constant $1$~lies in the ideal generated by the polynomials
corresponding to the clauses of~$\fstd$. Here the size of a proof is
measured as the number of monomials in a proof when all polynomials
are expanded out as linear combinations of monomials, and the width of
a clause corresponds to the (total) \introduceterm{degree} of the
polynomial representing the clause.
Briefly, the difference between PC and PCR is that the latter proof
system has separate formal variables for positive and negative
literals over the same variable. Thanks to this, one can encode
wide clauses into polynomials compactly regardless of the sign of the
literals in the clauses, which allows PCR to simulate resolution
efficiently. With respect to the degree measure 
polynomial calculus and polynomial calculus resolution
are exactly the same, and 
furthermore
the degree needed to prove 
in polynomial calculus
that a formula is unsatisfiable is
at most the width required in resolution.

In a work that served, interestingly enough, as a precursor
to~\cite{BW01ShortProofs}, Impagliazzo \etal~\cite{IPS99LowerBounds}
showed that strong lower bounds on the degree of PC proofs are
sufficient to establish strong size lower bounds. The same proof goes
through for PCR, 
and hence any lower bound on proof size obtained via a degree lower
bound applies to both PC and PCR.
In this paper, we will therefore
be somewhat sloppy in distinguishing the two proof systems, sometimes
writing ``polynomial calculus'' to refer to both systems when the
results apply to both PC and PCR.

In contrast to the situation for resolution
after~\cite{BW01ShortProofs}, the paper~\cite{IPS99LowerBounds}
has not been followed by a  corresponding
development of a generally applicable machinery for proving degree
lower bounds. 
For fields of characteristic distinct from~$2$ it is sometimes
possible to obtain lower bounds by doing an affine transformation from
$\set{0,1}$ to the ``Fourier basis'' $\set{-1,+1}$, an idea that seems
to have appeared first 
\ifthenelse{\boolean{conferenceversion}}
{in~\cite{BGIP01LinearGaps,Grigoriev98Tseitin}.}
{in~\cite{Grigoriev98Tseitin,BGIP01LinearGaps}.} 
For fields of arbitrary characteristic Alekhnovich and
Razborov~\cite{AR03LowerBounds} developed a powerful technique for
general systems of polynomial equations, which when restricted to the
standard encoding of CNF formulas~$F$ yields that
polynomial calculus
proofs require high
degree if the corresponding bipartite clause-variable incidence
graphs~$G(F)$ are good enough expanders.  
There are many formula families for which this is not true,
however. One can have a family of 
constraint satisfaction problems 
where the
constraint-variable incidence graph is an expander---say, for
instance, for an unsatisfiable set of linear equations $\bmod\, 2$---but
where each constraint is then translated into several clauses 
when encoded into CNF, 
meaning that the clause-variable incidence graph $G(F)$ will no longer
be expanding. For some formulas this limitation is inherent---it is
not hard to see that an inconsistent system of 
\ifthenelse{\boolean{conferenceversion}}
{linear equations $\bmod\, 2$ is easy to refute in polynomial calculus
  over~$\F_2$---but in other}
{linear equations $\bmod\, 2$ is easy to refute in polynomial calculus
  over~$\F_2$, and so good expansion for the constraint-variable
  incidence graph should \emph{not} in itself be sufficient to imply
  hardness in general---but in other}
cases it would seem that some kind of expansion of this sort should
still be enough, ``morally speaking,'' to
guarantee that the corresponding CNF formulas are hard.\footnote{In a bit more detail, what is shown in
  \cite{AR03LowerBounds} is that if the constraint-variable incidence
  graph for a set of polynomial equations is a good expander, and if
  these polynomials have high immunity---\ie do not imply other
  polynomials of significantly lower degree---then proving that this
  set of polynomial equations is inconsistent in polynomial calculus
  requires high degree. CNF formulas automatically have maximal
  immunity since a clause translated into a polynomial does not have
  any consequences of degree lower than the width of the clause in
  question, and hence expansion of the clause-variable incidence graph
  is sufficient to imply hardness for polynomial calculus. Any
  polynomial encoding of a linear equation $\bmod\, 2$ has a
  low-degree consequence over~$\F_2$, however---namely, the linear
  equation itself---and this is why \cite{AR03LowerBounds} (correctly)
  fails to prove lower bounds in this case.}

\subsection{Pigeonhole Principle Formulas}

One important direction in proof complexity, which is the reason
research in this area was initiated 
by Cook and Reckhow~\cite{CR79Relative}, 
has been to prove superpolynomial lower
bounds on proof size for increasingly stronger proof systems.  For
proof systems where such lower bounds have already been obtained,
\ifthenelse{\boolean{conferenceversion}}
{however, a somewhat orthogonal research direction has been to try to
  gain a better understanding of the strengths and weaknesses of the
  proof system by studying different combinatorial principles (encoded
  in CNF) and determining how hard they are to prove.}
{however, such as resolution and polynomial calculus,
  a somewhat orthogonal research direction has been to try to
  gain a better understanding of the strengths and weaknesses of a given
  proof system by studying different combinatorial principles (encoded
  in CNF) and determining how hard they are to prove
  for this proof  system.
}

It seems fair to say that by far the most extensively studied such
combinatorial principle is the \introduceterm{pigeonhole principle}.
This principle is encoded into CNF as unsatisfiable formulas claiming
that $m$~pigeons can be mapped in a one-to-one fashion into $n$~holes
for $m > n$, but there are several choices exactly how to do this
encoding. The most basic 
\introduceterm{pigeonhole principle (PHP) formulas}
have clauses saying that every pigeon gets at least one pigeonhole and that
no hole contains two pigeons. While these formulas are already
unsatisfiable for $m \geq n+1$, they do not a priori 
\ifthenelse{\boolean{conferenceversion}}
{rule out ``fat'' pigeons residing}
{rule out that there might be ``fat'' pigeons residing}
in several holes. The
\introduceterm{functional pigeonhole principle (FPHP) formulas} 
perhaps correspond more closely to our intuitive understanding of the
pigeonhole principle in that they also contain
\introduceterm{functionality} clauses specifying that every pigeon
gets exactly one pigeonhole 
and not more. Another way of making the basic PHP formulas 
more constrained
is to add \introduceterm{onto} clauses requiring that every
pigeonhole should get a pigeon,  yielding so-called
\introduceterm{onto-PHP formulas}. 
Finally, the most restrictive encoding, and hence the hardest one
when it comes to proving lower bounds, are the
\introduceterm{onto-FPHP formulas}
containing both functionality and onto clauses, \ie saying that the
mapping from pigeons to pigeonholes is a perfect matching.
Razborov's survey~\cite{Razborov02ProofComplexityPHP} gives a 
detailed account of these different flavours of the pigeonhole
principle formulas and results for them \wrt various proof
systems---we just quickly highlight some facts relevant to this paper
below.

For the resolution proof system there is not much need to distinguish
between the different PHP versions discussed above. 
The lower bound by
Haken~\cite{Haken85Intractability}
for formulas with $m=n+1$ pigeons can be made to work also for
onto-FPHP formulas, and more recent works by
Raz~\cite{Raz04Resolution} and
Razborov~\cite{Razborov03ResolutionLowerBoundsWFPHP,Razborov04ResolutionLowerBoundsPM}
show that the formulas remain exponentially hard (measured in the
number of pigeonholes~$n$) 
even
for arbitrarily many pigeons~$m$.

Interestingly enough, for polynomial calculus the story is very
different. The first degree lower bounds were proven by
Razborov~\cite{Razborov98LowerBound}, but for a different encoding
than the standard translation from CNF, since translating wide clauses
yields initial polynomials of high degree. 
Alekhnovich and Razborov~\cite{AR03LowerBounds}
proved lower bounds for a \mbox{$3$-CNF} version of the pigeonhole
principle, from which it follows that the standard CNF encoding
requires 
proofs of exponential size. 
However, as shown by Riis~\cite{Riis93Thesis} the onto-FPHP formulas
with $m=n+1$ pigeons are easy for polynomial calculus. 
And while the encoding in
\cite{Razborov98LowerBound}
also captures the functionality restriction in some sense, it has
remained open whether the standard CNF encoding of 
functional pigeonhole principle formulas
translated to polynomials is hard 
(this question has been highlighted, for instance, in Razborov's open
problems list~\cite{Razborov14webpage}).

\newcommand{\phpgraph}{H}

Another way of modifying the pigeonhole principle is to restrict the
choices of pigeonholes for each pigeon by defining the formulas over a
bipartite graph 
$\phpgraph = (U \disjointunion V, E)$ 
with 
$\setsize{U} = m$
and
$\setsize{V} = n$
and requiring that each pigeon $u \in U$ goes to one of its neighbouring
holes in 
$N(u) \subseteq V$. If the graph~$\phpgraph$ has constant left degree, the
corresponding \introduceterm{graph pigeonhole principle formula} has
constant width and a linear number of variables, which makes it
possible to apply \cite{BW01ShortProofs,IPS99LowerBounds} to obtain
exponential proof size lower bounds from linear width/degree lower
bounds. 
A careful reading of the proofs in~\cite{AR03LowerBounds} reveals that 
this paper establishes linear polynomial calculus degree lower bounds
(and hence exponential size lower bounds) for graph PHP formulas, and
in fact also graph Onto-PHP formulas, over constant-degree
expanders~$\phpgraph$.  Razborov lists as one of the open problems
in~\cite{Razborov02ProofComplexityPHP} whether this holds also for
graph FPHP formulas, \ie with functionality clauses added, from which
exponential lower bounds on polynomial calculus proof size for the
general FPHP formulas would immediately follow.

\subsection{Our Results}
\label{sec:our-results}

We revisit the technique developed in
\cite{AR03LowerBounds}
for proving polynomial calculus degree lower bounds, 
restricting our attention to the special case 
when the polynomials are obtained by
the canonical translation of CNF formulas.

Instead of considering the standard bipartite clause-variable
incidence graph~$G(F)$ 
of a CNF formula~$F$ 
(with clauses on the left, variables on
the right, and edges encoding that a variable occurs in a clause) we
construct a new graph~$G'$ by clustering several clauses and/or
variables into single 
vertices, reflecting the structure of 
\ifthenelse{\boolean{conferenceversion}}
{the encoded combinatorial principle.} 
{the combinatorial principle the CNF formula~$F$ is encoding.} 
The edges in this new graph~$G'$ are the ones
induced by the original graph~$G(F)$ in the natural way, \ie there is
an edge from a left cluster to a right cluster in~$G'$ if any clause in the
left cluster has an edge to any variable in the right cluster 
in~$G(F)$.
We remark that such a clustering is already implicit in, \eg, the
resolution lower bounds in~\cite{BW01ShortProofs} for Tseitin formulas
(which is essentially just a special form of unsatisfiable linear
equations) and graph PHP formulas, as well as in the graph PHP lower
bound for polynomial calculus in~\cite{AR03LowerBounds}.

\newcommand{\varsubset}{V}
\newcommand{\clsubset}{F'}
\newcommand{\varassnmt}{\rho}

We then show that 
if this clustering is done in the right way, 
the proofs in~\cite{AR03LowerBounds} still go through and yield strong
polynomial calculus degree lower bounds when~$G'$ is a good enough expander.\footnote{For a certain twist of the definition of expander that we do
  not describe in full detail here in order to keep the discussion
  at an informal, intuitive level. The formal description is given in
  \refsec{sec:method-cvig}.} 
It is clear that this cannot work in general---as already discussed above, 
any inconsistent system of linear equations $\bmod\, 2$ is easy to refute
in polynomial calculus over~$\F_2$, even though for a random instance
of this problem the clauses encoding each linear equation can be
clustered to yield an excellent expander~$G'$. Very informally (and
somewhat incorrectly) speaking, the clustering should be such that if
a cluster of clauses~$\clsubset$ on the left is a neighbour of a
variable cluster~$\varsubset$ on the right, then there should exist an 
assignment~$\varassnmt$ to~$\varsubset$ such that 
$\varassnmt$~satisfies all of~$\clsubset$ and such that 
for the clauses outside of $\clsubset$ they are either satisfied
by~$\varassnmt$ or left completely untouched by~$\varassnmt$. 
Also, it turns out to be helpful not to insist that the clustering of
variables on the right should be a partition, but that we should allow
the same variable to appear in several clusters if needed (as long as
the number of clusters for each variable is bounded).

This extension of the lower bound method in~\cite{AR03LowerBounds}
makes it possible to present previously obtained polynomial calculus
degree lower bounds in 
\cite{AR03LowerBounds,GL10Optimality,MN14LongProofs}
in a unified framework. Moreover, it allows us to prove the following
new results:
\begin{enumerate}
\item 
  If a bipartite graph  $H = (U \disunion V, E)$ 
  with
  $\setsize{U} = m$
  and
  $\setsize{V} = n$
  is a boundary expander (a.k.a.\ unique-neighbour expander),
  then the graph FPHP formula over~$\phpgraph$ requires
  proofs of linear polynomial calculus degree, and hence exponential
  polynomial calculus size.
  
\item 
  Since FPHP formulas can be turned into graph FPHP formulas by
  hitting them with a restriction, and since restrictions can only
  decrease proof size, it follows that FPHP formulas require  proofs
  of exponential size in polynomial calculus.
\end{enumerate}
This fills in the last missing pieces in our understanding of the
different flavours of pigeonhole principle formulas with 
\mbox{$n+1$ pigeons} and \mbox{$n$ holes} for polynomial
calculus. Namely, while Onto-FPHP formulas are easy for polynomial calculus,
both FPHP formulas and Onto-PHP formulas are hard even when restricted
to expander graphs.

\subsection{Organization of This Paper}

After reviewing the necessary preliminaries in
\refsec{sec:preliminaries}, we present our extension of the
Alekhnovich--Razborov method in \refsec{sec:method}.  
In \refsec{sec:applications}, we show how this method can be used to
rederive some previous 
polynomial calculus
degree lower bounds as well as to obtain new
degree and size lower bounds for functional (graph) PHP formulas. We
conclude in \refsec{sec:conclusion} by discussing some possible
directions for future research. \ifthenelse{\boolean{conferenceversion}}
{We refer to the upcoming full-length version for the details omitted
  in this extended abstract.}
{}

\makeatletter{}\section{Preliminaries}
\label{sec:preliminaries}

Let us start by giving 
an overview of the relevant proof complexity
background. This material is standard and we refer to,  for instance, the
survey~\cite{Nordstrom13SurveyLMCS} for more details.  

A
\introduceterm{literal} over a Boolean variable $\varx$ is either the
variable $\varx$ itself (a \introduceterm{positive literal}) or its
\mbox{negation $\lnot \varx$} or~$\olnot{\varx}$ (a 
\introduceterm{negative  literal}). We define $\olnot{\olnot{\varx}} = \varx$. 
We identify $0$ with true and $1$ with false.
We remark that this is the opposite of the standard convention in
proof complexity, but it is a more natural choice in the context of
polynomial calculus, where ``evaluating to true'' means ``vanishing.''
A
\introduceterm{clause}
$\clc = \lita_1 \lor \formuladots \lor \lita_{\clwidth}$ is a
disjunction of literals. A \introduceterm{CNF formula}
$\fstd = \clc_1 \land \formuladots \land \clc_m$ is a conjunction of
clauses. The \introduceterm{width}~$\widthofsmall{\clc}$ of a clause
$\clc$ is the number of literals~$\setsizesmall{\clc}$ 
in it, 
and the
width $\widthofsmall{F}$ of the formula $F$ is the maximum width of
any clause in the formula. We think of clauses and CNF formulas as
sets, so that order is irrelevant and there are no repetitions. A
$k$-CNF formula has all clauses of size at most~$k$, 
where $k$ is assumed to be some fixed constant.

In polynomial calculus resolution the goal is to prove the
unsatisfiability of a CNF formula by reasoning with polynomials from a
polynomial ring $\fieldstd[x, \olnot{x}, y, \olnot{y}, \ldots]$ (where
$x$ and $\olnot{x}$ are viewed as distinct formal variables) over some
fixed field~$\fieldstd$. The results in this paper hold for all
fields~$\fieldstd$ regardless of characteristic. 
In what follows, a \introduceterm{monomial}
$\monomialm$ is a product of variables and a \introduceterm{term}
$\termt$ is a monomial multiplied by an arbitrary non-zero field
element.
 
\begin{definition}[Polynomial calculus resolution (PCR) 
    \ifthenelse{\boolean{conferenceversion}}    {\cite{ABRW02SpaceComplexity,CEI96Groebner}}    {\cite{CEI96Groebner,ABRW02SpaceComplexity}}]
  A \introduceterm{polynomial calculus resolution (PCR) refutation}
  $\refof{\proofstd}{F}$ of a CNF formula~$F$ (also referred to as a
  \introduceterm{PCR proof} for~$F$) 
  over a field~$\fieldstd$
  is an ordered sequence of polynomials 
  $\proofstd = (\polyp_1, \ldots, \polyp_{\stoptime})$,
  expanded out as linear combinations of monomials, such that
  \mbox{$\polyp_{\stoptime} = 1$} and each line $\polyp_i$,
  $1 \leq i \leq \stoptime$, is either
  \begin{itemize}
  \item a monomial
    $\prod_{x \in \litpos} x \cdot \prod_{y \in \litneg} \olnot{y}$
    encoding a clause
    $\Lor_{x \in \litpos} x \lor \Lor_{y \in \litneg} \olnot{y}$
    in~$\fstd$
    (a \introduceterm{clause axiom});
  
  \item a \introduceterm{Boolean axiom} $x^2 - x$ or
    \introduceterm{complementarity axiom} $x + \olnot{x} - 1$ for any
    variable~$x$;

  \item a polynomial obtained from one or two
    \ifthenelse{\boolean{conferenceversion}} 
    {previous polynomials by}
    {previous polynomials in the sequence by} \introduceterm{linear
      combination}
    $\frac{\polyq \quad \polyr}{\alpha \polyq + \beta \polyr}$ or
    \introduceterm{multiplication} $\frac{\polyq}{x \polyq}$ for any
    $\alpha, \beta \in \fieldstd$ and any variable~$x$.
  \end{itemize}
  If we drop complementarity axioms and encode each negative literal
  $\olnot{x}$ as~  $(1 - x)$, the proof system is called
  \introduceterm{polynomial calculus (PC)}.

  The \introduceterm{size} $\sizeofarg{\proofstd}$ of a PC/PCR
  refutation
  $\proofstd = (\polyp_1, \ldots, \polyp_{\stoptime})$
  is the number of monomials  in~$\proofstd$
  (counted with repetitions),  \footnote{We remark that the natural definition of size is to count
    monomials with repetition, but all lower bound techniques known
    actually establish slightly stronger lower bounds on the number of
    \emph{distinct} monomials.}
  the \introduceterm{degree}
  $\mdegreeof{\proofstd}$ is the maximal degree of any monomial
  appearing in~$\proofstd$, and the \introduceterm{length}
  $\lengthofarg{\proofstd}$ is the number $\stoptime$ of polynomials in
  $\proofstd$. Taking the minimum over all PCR refutations of a
  formula~$F$, we define the size $\sizeref[\pcrnot]{F}$, degree
  $\mdegreeref[\pcrnot]{F}$, and length $\lengthref[\pcrnot]{F}$ of
  refuting~$F$ in PCR (and analogously for PC).
\end{definition}

We write $\vars{\clc}$ and
$\vars{\monomialm}$ to denote the set of all variables appearing in a
clause~$\clc$ or monomial (or term)~$\monomialm$, respectively and extend
this notation to CNF formulas and polynomials by taking unions.
We use the notation $\idealgen{\polyp_1, \ldots, \polyp_m}$ for
the ideal generated by the 
\mbox{polynomials  $\polyp_i$, $i \in [m]$}.
That is,
$\idealgen{\polyp_1, \ldots, \polyp_m}$ is the minimal subset of
polynomials containing 
all~$\polyp_i$
that is closed under addition and
multiplication by any polynomial. One way of viewing 
a
polynomial calculus (PC or PCR) refutation is 
as a calculation
in the
ideal generated by the encodings of clauses in $\formf$
and the Boolean and complementarity axioms.
It can be shown that such an ideal contains~$1$ 
if and only if $\formf$ is 
unsatisfiable. 

As mentioned above, we have
$\mdegreeref[\pcrnot]{F} = \mdegreeref[\pcnot]{F}$ for any
CNF formula $F$. 
This claim can essentially be verified by taking any PCR refutation
of~$\formf$ 
and replacing all occurrences of~$\olnot{y}$
by~$(1-y)$ to obtain a valid PC refutation in the same degree. 
Hence, we can drop the subscript from the notation for the degree
measure. We have the following relation between refutation size and
refutation degree (which was originally proven for~PC but the proof of
which also works for~PCR).

\begin{theorem}[\cite{IPS99LowerBounds}]
  \label{th:IPSPCTheorem}
  Let $\formf$ be an unsatisfiable CNF formula of width
  $\widthofarg{\formf}$ over $n$~variables. Then
  \begin{equation*}
    \sizeref[\pcrnot]{\formf} = 
    \exp\left(\BIGOMEGA{
        \frac{\left(\mdegreeref{\formf} 
            - \widthofarg{\formf}\right)^2}
        {n}
      }\right)
    \eqperiod
  \end{equation*}
\end{theorem}

Thus,  for $k$-CNF formulas it is sufficient to prove strong enough
lower bounds on the PC degree of refutations to establish strong lower
bounds on PCR proof size.

Furthermore, it will be convenient for us to simplify the definition
of~PC so that axioms $x^2 - x$ are always applied implicitly whenever
possible. We do this by defining the result of 
the multiplication operation to be the multilinearized version of the
product. This can only decrease the degree (and size) of the
refutation, and is in fact how polynomial calculus is defined
in~\cite{AR03LowerBounds}. 
Hence, from now on whenever we refer to
polynomials and monomials we mean multilinear polynomials and
multilinear monomials, respectively, and polynomial calculus is
defined over the (multilinear) polynomial 
\mbox{ring
$\fieldstd[x, y, z, \ldots] / \idealgen{x^2-x, y^2-y, z^2-z,
  \ldots}$.}

\ifthenelse{\boolean{conferenceversion}}
{}
{It might be worth noticing that for this modified definition of
  polynomial calculus it holds that any (unsatisfiable) $k$-CNF formula
  can be refuted in linear length (and hence, in constrast to
  resolution, the size of refutations, rather than the length, is the
  right measure to focus on).
  This is not hard to show, and in some sense is probably folklore,
  but since it does not seem to be too widely known we
  state it for the record and  provide a proof.
  
  \begin{proposition}
    Any unsatisfiable $k$-CNF formula $F$
            has a
    (multilinear) polynomial calculus refutation 
    of
        length linear in
    the size of the formula~$F$.
  \end{proposition}

  \begin{proof}
                We show by induction how to derive polynomials
    $\polyp_j = 1 - \prod_{i = 1}^{j} (1 - \polyax_i)$ 
    in length linear  in~$j$, 
            where we identify the clause~$\polyax_i$ in 
    $F = \Land_{i=1}^{m} \polyax_i$ with the polynomial
    encoding of this clause.
    The end result is the polynomial 
    $\polyp_m = 1 - \prod_{i=1}^{m}(1 - \polyax_i)$. 
            As $F$ is unsatisfiable, for every
    $0$-$1$ assignment there is at least one $\clc_i$ that evaluates
    to~$1$ and hence $\polyp_m$ evaluates to $1$. Thus, $\polyp_m$ is
    equal to $1$ on all $0$-$1$ assignments.
            However, it is a basic fact that every function
    $\funcdescr{f}{\set{0,1}^n}{\F}$
    is uniquely representable as a multilinear polynomial in
    $\F[x_1, \ldots, x_n]$
    (since the multilinear monomials span this vector space and are
    linearly independent, they form a basis).
    Therefore, 
    it follows that $\polyp_m$ is syntactically equal to the
    polynomial~$1$. 
    
    The base case of the induction is the polynomial $\polyp_1$ that is
    equal to $\polyax_1$. To prove the induction step, 
            we need to show how to derive 
    \begin{equation}
      \polyp_{j + 1} 
      = 
      1 - \prod_{i = 1}^{j+1} (1 - \polyax_i) 
      =
      1 - (1 -  \polyax_{j+1}) (1 - \polyp_j)
      = 
      \polyp_j + \polyax_{j+1} - \polyax_{j+1} \polyp_j
    \end{equation}
    from $\polyp_j$ and $\polyax_{j + 1}$ in
    a constant number of steps.
    To start, we derive $\polyax_{j + 1} \polyp_j$ from
    $\polyp_j$, which can be done with a constant number of
    multiplications and additions since the width/degree
    of $\polyax_{j+1}$ is upper-bounded by the constant~~$k$.
    We derive $\polyp_{j + 1}$ in two more steps by first taking a
    linear combination
    of   $\polyp_j$ and $\polyax_{j + 1} \polyp_j$ to get
    $\polyp_j - \polyax_{j + 1} \polyp_j$
    and then adding $\polyax_{j+1}$ to this to obtain
    $\polyp_j - \polyax_{j + 1} \polyp_j + \polyax_{j + 1} = \polyp_{j+1}$.
    The proposition follows.
                                          \end{proof}}

We will also need to use restrictions.
A \introduceterm{restriction}
$\rstd$ on $\fstd$ is a partial assignment to the variables
of~$\fstd$. We use $\domainof{\rstd}$ to denote the set of variables
assigned by $\rstd$. In a restricted formula $\restrict{\fstd}{\rstd}$
all clauses satisfied by $\rstd$ are removed and all other clauses
have falsified literals removed. 
For a PC refutation~$\proofstd$ restricted by~$\rstd$ we have that
if~$\rstd$ satisfies a literal in a monomial, then that monomial is
set to~$0$ and vanishes, and all falsified literals in a monomial get
replaced by~$1$ and 
disappear.
It is not hard to see that if~$\proofstd$
is a PC (or PCR) refutation of~$\fstd$, then $\restrict{\proofstd}{\rstd}$ is a
PC (or PCR) refutation of $\restrict{\fstd}{\rstd}$, and  this restricted
refutation has at most the same size, degree, and length as the
original refutation.

\makeatletter{}\section{A Generalization of the Alekhnovich--Razborov Method for CNFs} 
\label{sec:method}

Many lower bounds in proof complexity are proved by arguing in terms
of expansion.  One common approach is to associate a bipartite graph
$G(\formf)$ with the CNF formula~$\formf$ with clauses on one side and
variables on the other and with edges encoding that a variable occurs
in a clause (the so-called \introduceterm{clause-variable incidence graph} 
mentioned in the introduction).  The method we present below,
which is an extension of the techniques developed by Alekhnovich and
Razborov~\cite{AR03LowerBounds} (but restricted to the special case of
CNF formulas), is a variation on this theme.  As already discussed,
however, we will need a slightly more general graph construction where
clauses and variables can be grouped into clusters.  We begin by
describing this construction.

\subsection{A Generalized Clause-Variable Incidence Graph}
\label{sec:method-cvig}

The key to our construction of generalized clause-variable incidence
graphs is to keep track of how clauses in a CNF formula are affected
by partial assignments.

\begin{definition}[\Rsatrespectful assignments and variable sets]
  \label{def:respectful-assignments}
  We say that
  a partial assignment $\assignmenta$ 
  \introduceterm{\rsatrespects{}} a CNF formula~$\fixedconstraints$,
  or that~$\assignmenta$ is \introduceterm{\radjective{}},
 if for every clause $\clc$ in
  $\fixedconstraints$ either
  $\vars{\clc} \intersection \domainof{\assignmenta} = \emptyset$ or
  $\assignmenta$ satisfies $\clc$. 
  A set of variables $\variablesv$
  \rsatrespects a CNF formula $\fixedconstraints$ if there exists 
  an assignment~$\assignmenta$ with
  $\domainof{\assignmenta} = \variablesv$ that respects
  $\fixedconstraints$.
\end{definition}

\begin{example}
  Consider the CNF formula
  $
  \fixedconstraints = 
  (\varx_1 \land \varx_2)
  \land
  (\olnot{\varx}_1 \land \varx_3)
  \land
  (\varx_1 \land \varx_4)
  \land
  (\olnot{\varx}_1 \land \varx_5)
  $
  and the subsets of variables
  $\variablesv_1 = \set{\varx_1, \varx_2, \varx_3}$
  and
  $\variablesv_2 = \set{\varx_4, \varx_5}$.
  The assignment $\assignmenta_2$ to~$\variablesv_2$ setting
  $\varx_4$ and~$\varx_5$ to true
  \rsatrespects~$\fixedconstraints$ since it
  satisfies the clauses containing
  these variables, and hence $\variablesv_2$ is \radjective.
  However, $\variablesv_1$~is not \radjective since setting~$\varx_1$
  will affect all clauses in~$\fixedconstraints$ but cannot satisfy both
  $\varx_1 \land \varx_4$
  and
  $\olnot{\varx}_1 \land \varx_5$.
\end{example}

\begin{definition}[\Rsatrespectful satisfaction]
  \label{def:respectful-satisfaction}
  Let~$\formf$ and~$\fixedconstraints$ be CNF formulas and
  let~$\variablesv$ be a set of variables.
  We say that $\formf$ is
  \introduceterm{\rsatisfiable by $\variablesv$}
  if there exists a partial
  assignment~$\assignmenta$ with 
  $\domainof{\assignmenta} = \variablesv$ that  
  satisfies~$\formf$ and 
  respects~$\fixedconstraints$. Such an assignment~$\assignmenta$ is
  said to  \introduceterm{\rsatisfy{}} $\formf$.
\end{definition}

Using a different terminology,
\refdef{def:respectful-assignments} 
says that 
$\assignmenta$ is an \introduceterm{autarky} for $\fixedconstraints$, 
meaning that
$\assignmenta$ satisfies all clauses in~$\fixedconstraints$ which it
touches,
\ie that
$\restrict{\fixedconstraints}{\assignmenta} \subseteq
\fixedconstraints$
after we remove all satisfied clauses in
$\restrict{\fixedconstraints}{\assignmenta}$.
\Refdef{def:respectful-satisfaction} ensures that the autarky
$\assignmenta$ satisfies the formula $\formf$.

Recall that we identify a CNF formula 
$\Land_{i=1}^{m} \clc_i$
with the set of clauses
$\setdescr{\clc_i}
{i \in [m]}
$.
In the rest of this section we will switch freely between
these two perspectives.
We also change
to the notation
$\originalcnf$ for the input CNF formula, to free up other letters
that will be needed in notation introduced below.

To build a bipartite graph representing the CNF formula $\originalcnf$, we 
will group the formula into subformulas (\ie subsets of
clauses). In what follows, we write
$\fpartition$ 
to denote the part of $\originalcnf$ that will form the left
vertices of the constructed
bipartite graph, while 
$\fixedconstraints$ denotes the part of $\originalcnf$ which
will not be represented in the graph
but will be used to enforce \rsatrespectful satisfaction.
In more detail, $\fpartition$ is a family of subformulas $\formf$ 
of $\originalcnf$ where each subformula is one vertex on the left-hand
side of the graph.
We also consider
the variables of $\originalcnf$
to be divided into a family $\vpartition$ of 
subsets of variables~$\variablesv$. 
In our definition, $\fpartition$~and $\vpartition$ do not need
to be partitions of clauses and variables in $\originalcnf$, respectively. 
This is not too relevant for~$\fpartition$ because we will always
define it as a partition,  
but it turns out to be useful in our applications
to have sets in $\vpartition$ share variables.
The next definition describes the bipartite graph that we build and 
distinguishes between two types of 
neighbour relations
in this graph.

\begin{definition}[Bipartite \fvstructure{}]
  \label{def:f-v-graph}
  Let $\fixedconstraints$ be a CNF formula, $\fpartition$ be a set of
  CNF formulas, and $\vpartition$ be a family of sets of variables
  $\variablesv$ 
  that \rsatrespect $\fixedconstraints$.
  Then the 
  \introduceterm{(bipartite) \fvstructure} 
  is a bipartite graph with left vertices
  $\formf \in \fpartition$, right vertices
  $\variablesv \in \vpartition$, 
  and edges between $\formf$ and $\variablesv$ if
  \mbox{$\vars{\formf} \intersection \variablesv \neq \emptyset$}.
  For every edge $(\formf, \variablesv)$ in the graph we
  say that $\formf$ and $\variablesv$ are
  \introduceterm{\rneighbours{}} if~$\formf$ is \rsatisfiable
  by~$\variablesv$. Otherwise, they are
  \introduceterm{\rnegneighbours{}}.
\end{definition}

We will often write
$\fvstructurenot$
as a shorthand for the graph defined by
$\fpartition$, $\vpartition$, and $\fixedconstraints$
as above.
We will also  use standard graph notation and write
$\neighbor{\formf}$ to denote the set of all neighbours
$\variablesv \in \vpartition$ of 
a vertex/CNF formula
$\formf \in \fpartition$.
It is important to note that the fact that $\formf$ and~$\variablesv$
are \rneighbours can be witnessed by  an 
assignment that falsfies other subformulas 
$\formf' \in \fpartition \setminus \set{\formf}$.

\ifthenelse{\boolean{conferenceversion}}
{}
{We can view the formation of the \fvstructure as taking the
  clause-variable incidence graph
  $G(\originalcnf)$ of 
      the
  CNF formula $\originalcnf$, throwing out a part of
  $\originalcnf$, which we denote $\fixedconstraints$, and clustering the
  remaining clauses and variables into $\fpartition$ and
  $\vpartition$. The edge relation in the \fvstructure follows
  naturally from this view, as we put an edge between two clusters if
  there is an edge between any two elements of these clusters. The only
  additional information we need to keep track of is 
      which clause and variable clusters are
  \rneighbours or not.}
  
\begin{definition}[\Rsatrespectful boundary]
  \label{def:respectful-boundary}
  For 
  \afvstructure and a subset
  $\fpartition' \subseteq \fpartition$,
  the
  \introduceterm{\rboundary}
  $\sboundary{\fpartition'}$ of~$\fpartition'$
  is the family of variable
  sets~$\variablesv \in \vpartition$ such that each
  $\variablesv \in \sboundary{\fpartition'}$ is
  \anrneighbour
  of some clause set
  $\formf \in \fpartition'$ but is not a neighbour 
  (\rsatrespectful or \rsatdisrespectful)
  of any  other clause set 
  $\formf' \in \fpartition' \setminus \set{\formf}$.
\end{definition}

It will sometimes be convenient to interpret subsets 
$\fpartition' \subseteq \fpartition$ 
\ifthenelse{\boolean{conferenceversion}}
{as formulas} 
{as CNF formulas} 
$
\Land_{\formf \in \fpartition'} 
\Land_{\clc \in \formf} \clc
$,
and we will switch back and forth between these two interpretations as
seems most suitable.
We will show that 
a formula 
$
\originalcnf  = 
\Land_{\formf \in \fpartition} 
\Land_{\clc \in \formf} \clc \land \fixedconstraints
= \fpartition \land \fixedconstraints$ 
is hard for polynomial calculus \wrt degree if the
\fvstructure has a certain expansion property as defined next.

\begin{definition}
  [\Rsatrespectful boundary expander]
  \label{def:F-V-boundary-expansion}
  \Afvstructure is 
  said to be
  an
  \introduceterm{\rboundaryexpanderstd{}},
  or just an
  \introduceterm{\rboundaryexpandershortstd{}} for brevity,
  if for every set
  $\fpartition' \subseteq \fpartition$, 
  $\setsize{\fpartition'} \leq  \expsize$,
  it holds that
  $\setsize{\sboundary{\fpartition'}} \geq 
  \expfactor
  \setsize{\fpartition'} - \expslack$.
\end{definition}

\ifthenelse{\boolean{conferenceversion}}{}{Note that an
\rboundaryexpanderstd
is a standard bipartite boundary expander except for two modifications:
\begin{itemize}
\item We measure expansion not in terms of the whole boundary but only
  in terms of the \emph{\rsatrespectful boundary}  \footnote{Somewhat intriguingly, we will not see any
    \rsatdisrespectful neighbours in our applications in
    \refsec{sec:applications}, but the concept of \rsatrespectfulness
    is of crucial importance for the main technical result in
    \refth{th:FVStructureTheorem} to go through.
    One way of seeing this is to construct
    \afvstructure for an expanding set of linear equations
    $\bmod\, 2$, where $\fpartition$ consists of the (CNF encodings of)
    the equations, $\vpartition$ consists of one variable set for
    each equation containing exactly the variables in this equation,
    and $\fixedconstraints$ is empty. Then this \fvstructure has the
    same boundary expansion as the constraint-variable incidence
    graph, but \refth{th:FVStructureTheorem} does not apply (which it
    should not do) since this expansion is not \rsatrespectful.}
  as described in \refdef{def:respectful-boundary}.

\item 
  Also, the size of the boundary~$\setsize{\sboundary{\fpartition'}}$
  on the right does not quite have to scale linearly
  with the size of the vertex set
  $\setsize{\fpartition'}$
  on the left. 
  Instead, we allow an
  \emph{additive loss}~$\expslack$ in the expansion.
  In our applications, we can usually construct graphs
  with good enough expansion so that
  we can choose $\expslack=0$, but for one of the results we present
  it will be helpful to allow a small slack here.
\end{itemize}}

Before we state our main theorem we need one more technical
definition, 
which is used to ensure
that there do not exist
variables that appear in too many variable sets in $\vpartition$.
\ifthenelse{\boolean{conferenceversion}}
{}
{We remark that the concept below is also referred to as the ``maximum
  degree'' in the literature, but since we already have
    degrees of polynomials and vertices in this paper
  we prefer a new term instead of overloading ``degree'' with a third
  meaning.}

\begin{definition}
  \label{def:overlap}
  The \introduceterm{\voverlapterm{}} of a variable $x$ \wrt a family
  of variable sets $\vpartition$ is
  $ \xoverlapnot{x}{\vpartition} = \setsize{\setdescr[:]{\variablesv
      \in \vpartition}{x \in \variablesv}} $
  and the \voverlapterm of $\vpartition$ is
  $ \voverlapnot{\vpartition} =
  \maxofexpr[x]{\xoverlapnot{x}{\vpartition}} $,
  \ie the maximum number of sets $\variablesv \in \vpartition$
  containing any particular variable $\varx$.
\end{definition}

Given the above definitions, we can state the main technical result in
this paper as follows.

\begin{theorem}
  \label{th:FVStructureTheorem}
  Let
  $
  \originalcnf  = 
  \Land_{\formf \in \fpartition} 
  \Land_{\clc \in \formf} \clc \land \fixedconstraints
  = \fpartition \land \fixedconstraints$ 
  be a CNF formula for which 
  $\fvstructurenot$ 
  is an \rboundaryexpandershortstd{} with 
  \voverlapterm   $ \voverlapnot{\vpartition} = \setdegd $,
  and suppose furthermore that for all
  $\fpartition' \subseteq \fpartition, \, \setsize{\fpartition'} \leq
  \expsize$,
  it holds that $\fpartition' \land \fixedconstraints$ is
  satisfiable. Then any polynomial calculus 
  refutation of~$\originalcnf$
  requires degree
  strictly greater than
  $(\expfactor \expsize - 2 \expslack) / (2 \setdegd)$.
\end{theorem}

In order to prove 
this theorem,
it will be convenient to 
review some algebra. We do so next.

\subsection{Some Algebra Basics}
\label{sec:pc-algebra-basics}

We will need to compute with polynomials modulo ideals, and in order
to do so we need to have an ordering of monomials (which, as we
recall, will always be multilinear).

\begin{definition}[Admissible ordering]
  \label{def:admissible-ordering}
  We say that a total ordering $\mlt$ on the set of all monomials
  over some fixed set of variables
  is
  \introduceterm{admissible} if the following conditions hold:
  \begin{itemize}
  \item 
    If
    $\mdegreeof{\monomialm_1} < \mdegreeof{\monomialm_2}$, 
    then
    $\monomialm_1 \mlt \monomialm_2$.
  \item 
    For any
    $\monomialm_1, \monomialm_2$, and $\monomialm$
    such that
    $\monomialm_1 \mlt \monomialm_2$ and 
    \mbox{$\vars{\monomialm}
      \intersection
      \bigl( \vars{\monomialm_1} \union
      \vars{\monomialm_2} \bigr) = \emptyset$},
    it holds that 
    $\monomialm \monomialm_1 \mlt \monomialm \monomialm_2$.
  \end{itemize}
  Two terms
  $\termt_1 = \alpha_1 \monomialm_1$
  and
  $\termt_2 = \alpha_2 \monomialm_2$
  are ordered in the same way as their underlying monomials
  $\monomialm_1$
  and~$\monomialm_2$.
\end{definition}

One example of an admissible ordering is to first order monomials with
respect to their degree  
\ifthenelse{\boolean{conferenceversion}}
{and then lexicographically.} 
{and then lexicographically.
  For the rest of
  this section we only need that $\mlt$ is some fixed but arbitrary
  admissible ordering, but the reader can think of the
  degree-lexicographical ordering without any particular loss of
  generality.} 
We write
$\monomialm_1 \mleq \monomialm_2$
to denote that
$\monomialm_1 \mlt \monomialm_2$
or
$\monomialm_1 = \monomialm_2$.

\begin{definition}[Leading, reducible, and irreducible terms]
  For a polynomial 
  $\polyp = \sum_i \termt_i$,  
  the \introduceterm{leading term} 
  $\leadingterm{\polyp}$
  of~$\polyp$
  is the largest term~$\termt_i$ according to~$\mlt$.
  Let 
  $\ideali$
  be an ideal over the (multilinear) polynomial ring
  $\fieldstd[x, y, z, \ldots] / 
  \idealgen{x^2-x, y^2-y, z^2-z, \ldots}$.
  We say that a term~$\termt$ is
  \introduceterm{reducible modulo $\ideali$} if there exists a
  polynomial $\idpolyq \in \ideali$ such that 
  $\termt = \leadingterm{\idpolyq}$
and that $\termt$ is \introduceterm{irreducible modulo~$\ideali$}
otherwise.  
\end{definition}

The following fact is not hard to verify.

\begin{fact}
  \label{fact:unique-sum}
  Let 
  $\ideali$
  be an ideal over 
  $\fieldstd[x, y, z, \ldots] / 
  \idealgen{x^2-x, y^2-y, z^2-z, \ldots}$.
  Then any multilinear polynomial 
  $\polyp
  \in
  \fieldstd[x, y, z, \ldots] / 
  \idealgen{x^2-x, y^2-y, z^2-z, \ldots}$
  can be written uniquely as
  a sum $\idpolyq + \redpolyr$, where $\idpolyq \in \ideali$ and  
  $\redpolyr$ is a linear combination of irreducible terms
  modulo~$\ideali$. 
\end{fact} 

This 
is what allows us to reduce polynomials modulo an ideal in a
well-defined manner.

\ifthenelse{\boolean{conferenceversion}}{\begin{definition}[Reduction operator]
  \label{def:reduction-operator}
  Let $\ideali$ be an ideal and let $\polyp$ be any multilinear
  polynomial over
  $\fieldstd[x, y, z, \ldots] / \idealgen{x^2-x, y^2-y, z^2-z, \ldots}$. 
  The \introduceterm{reduction operator}~  $\redophead_{\ideali}$ is the operator that when applied to $\polyp$
  returns the sum of irreducible terms
  $\redop[\ideali]{\polyp} = \redpolyr$ such that
  $\polyp - \redpolyr \in \ideali$.
\end{definition}

}{\begin{definition}[Reduction operator]
  \label{def:reduction-operator}
  Let
  $\polyp \in
  \fieldstd[x, y, z, \ldots] / 
  \idealgen{x^2-x, y^2-y, z^2-z, \ldots}$
  be any multilinear polynomial
  and let
  $\ideali$
  be an ideal over 
  $\fieldstd[x, y, z, \ldots] / 
  \idealgen{x^2-x, y^2-y, z^2-z, \ldots}$.
  The
  \introduceterm{reduction operator}~  $\redophead_{\ideali}$ 
  is the operator that when applied to $\polyp$
  returns the sum of irreducible terms 
  $\redop[\ideali]{\polyp} = \redpolyr$ such that
  $\polyp - \redpolyr \in \ideali$.
\end{definition}
}

We conclude our brief algebra review by stating
two observations that are more or less immediate, but
are helpful enough for us to want to highlight them explicitly.

\begin{observation}
  \label{obs:reduction-mod-larger-ideal}
  For any two ideals
  $\ideali_1$, $\ideali_2$ 
  such that
  $\ideali_1 \subseteq \ideali_2$
  and any two polynomials $\polyp$, $\polyp'$ it holds that
  $\redop[\ideali_2]{\polyp \cdot \redop[\ideali_1]{\polyp'}} =
  \redop[\ideali_2]{\polyp \polyp'}$.
\end{observation}

\ifthenelse{\boolean{conferenceversion}}{}{
\begin{proof}
  Let
  \begin{equation}
    \polyp' = \idpolyq' + \redpolyr'
  \end{equation}
  for $\idpolyq' \in \ideali_1$
  and
  $\redpolyr'$ a linear combination of irreducible terms over~$\ideali_1$.
  Let
  \begin{equation}
    \polyp \cdot \redop[\ideali_1]{\polyp'} = 
    \polyp\redpolyr' = \idpolyq + \redpolyr
  \end{equation}
  for $\idpolyq \in \ideali_2$
  and
  $\redpolyr$ a linear combination of irreducible terms
  over~$\ideali_2$.
  Then
  \begin{equation}
    \polyp\polyp' =
    \polyp\idpolyq' + \polyp\redpolyr' =
    \polyp\idpolyq' + \idpolyq + \redpolyr
  \end{equation}
  where
  $\polyp\idpolyq' + \idpolyq \in \ideali_2$.
  By the uniqueness in
  \reffact{fact:unique-sum},
  we conclude that
  the equality
  $
  \redop[\ideali_2]{\polyp \polyp'} =
  \redpolyr = 
  \mbox{$\redop[\ideali_2]{\polyp \cdot \redop[\ideali_1]{\polyp'}}$}$
  holds.
\end{proof}
}

\begin{observation}
  \label{obs:restrictions-preserve-irreducibility}
  Suppose that the term
  $\termt$ 
  is irreducible modulo the ideal~$\ideali$
  and let $\assignmenta$ be 
  any partial assignment of variables in $\vars{\termt}$
  to values in~$\fieldstd$
  such that
  $\restrict{\termt}{\assignmenta} \neq 0$.
  Then
  $\restrict{\termt}{\assignmenta}$
  is also irreducible modulo~$\ideali$.
\end{observation}

\ifthenelse{\boolean{conferenceversion}}{}{
\begin{proof}
  Let 
  $
  \monomialm_{\assignmenta}
  $ 
  be the product of all variables in~$\termt$ assigned
  by~$\assignmenta$
  and let
  $\alpha = \restrict{\monomialm_{\assignmenta}}{\assignmenta}$,
  where by assumption we have $\alpha \neq 0$.
  If there is a polynomial~$\idpolyq \in \ideali$
  such that
  $\leadingterm{\idpolyq} = 
  \restrict{\termt}{\assignmenta}$,
  then 
  $\alpha^{-1} \monomialm_{\assignmenta} \idpolyq \in \ideali$
  and
  $\leadingterm{\alpha^{-1} \monomialm_{\assignmenta} \idpolyq} = \termt$,
  contradicting that~$\termt$ is irreducible.  
\end{proof}
}

\subsection{Proof Strategy}

Let us now state the lemma on which we base the proof of
\refth{th:FVStructureTheorem}.

\begin{lemma}[\cite{Razborov98LowerBound}]
  \label{lem:PCDegreeLowerBounds}
  Let $\originalcnf$ be any CNF formula and
  $\degd \in \Nplus$ be a positive integer. 
  Suppose that there exists a linear operator $\redophead$ on
  multilinear polynomials over $\vars{\originalcnf}$
  with the following properties:
  \begin{enumerate}
  \item 
    \label{item:R-one}
    $\redop{1} \neq 0$.
  \item
    \label{item:R-two}
    $\redop{\polyax} = 0$ for 
    (the translations to polynomials of)
    all axioms $\polyax \in \originalcnf$.
  \item
    \label{item:R-three}
    For every term $\termt$ with $\mdegreeof{\termt} < \degd$ and
    every variable $\varx$ it holds that
    $\redop{\varx \termt} = \redop{\varx \redop{\termt}}$.
  \end{enumerate}
  Then any polynomial calculus refutation of~$\originalcnf$ 
  (and hence any PCR refutation  of~$\originalcnf$)
  requires degree
  strictly greater than $\degd$.
\end{lemma}

\ifthenelse{\boolean{conferenceversion}}{}{The proof of
\reflem{lem:PCDegreeLowerBounds}
is not hard. The basic idea is that
$\redophead$ will map all axioms to~$0$ by 
property~\ref{item:R-two}, and further derivation steps in degree at
most~$\degd$ will yield polynomials that also map to~$0$ by 
property~\ref{item:R-three}
and the linearity of~$\redophead$. But then
property~\ref{item:R-one} implies that no derivation
in degree at most~$\degd$ can reach contradiction. }

To prove \refth{th:FVStructureTheorem}, we construct a linear operator
$\fvredophead$ that satisfies the conditions of \reflem{lem:PCDegreeLowerBounds}
when the \fvstructure $\fvstructurename$ is an expander. First, let us 
describe how we make the connection between polynomials and
the given \fvstructure.
We remark that
in the rest of this section we will identify a clause~$\clc$ 
with its polynomial translation and will
refer to $\clc$ as a (polynomial) axiom. 

\begin{definition}[Term and polynomial neighbourhood]
  \label{def:term-neighbourhood}
  The \introduceterm{neighbourhood} $\neighbor{\termt}$ of a term~$\termt$
  \wrt~$\fvstructurenot$  is 
  $
  \neighbor{\termt} = 
  \setdescr
  {\variablesv \in \vpartition}
  {\vars{\termt} \intersection \variablesv \neq \emptyset}
  $,
  \ie the family of all variable sets containing variables mentioned
  by~$\termt$.
  The neighbourhood of a polynomial 
  $\polyp = \sum_i \termt_i$ 
  is 
  $
  \neighbor{\polyp}
  =
  \Union_i \neighbor{\termt_i}
  $,
  \ie the union of   the neighbourhoods of all terms in~$\polyp$. 
\end{definition}

To every polynomial we can now assign a family of variable sets
$\vpartition'$.  But we are interested in the axioms that are needed
in order to produce that polynomial. That is, given a family of
variable sets~$\vpartition'$, we would like to identify the largest
set of axioms $\fpartition'$ that could possibly have been 
used in a derivation that yielded polynomials~$\polyp$ with
$\vars{\polyp} \subseteq \Union_{\variablesv \in \vpartition'}
\variablesv$.  This is the intuition behind the next definition.\footnote{We remark that \refdef{def:support} is a slight modification
  of the original definition of support in~\cite{AR03LowerBounds} that was
  proposed by Yuval
  Filmus~\cite{Filmus14OnAlekhnovichRazborovDegreeLowerBound}.} 

\begin{definition}[Polynomial support]
  \label{def:support}
  For a given \fvstructure and a family of variable sets
  $\vpartition' \subseteq \vpartition$, we say that a subset
  $\fpartition' \subseteq \fpartition$ is
  \introduceterm{\svcontainedstd{}}
  if 
  $\setsize{\fpartition'} \leq \expsize$
  and
  $\sboundary{\fpartition'} \subseteq  \vpartition'$.

  We define the
  \introduceterm{polynomial \ssupport{} $\support{\vpartition'}$ 
    of $\vpartition'$ \wrt~$\fvstructurenot$},
  or just
  \introduceterm{\ssupport{} of~$\vpartition'$} for brevity,
  to be
  the union of all \svcontainedstd
  subsets $\fpartition' \subseteq \fpartition$,
  and the \ssupport $\support{\termt}$
  of a term~$\termt$ is 
  defined to be
  the \ssupport of~$\neighbor{\termt}$.
\end{definition}

We will usually just speak about ``support'' below without further
qualifying this term, 
since the \fvstructure~$\fvstructurename$ 
will be clear from context.
The next observation follows immediately from
\refdef{def:support}.

\begin{observation}
  \label{obs:monotone-support}
  Support
  is monotone in the sense that
  if
  $\termt \subseteq \termt'$
  are two terms, then
  it holds that
  $\support{\termt}\subseteq \support{\termt'}$.
\end{observation}

Once we have identified the axioms that are potentially involved in
deriving $\polyp$, we define the linear operator~$\fvredophead$ as
the
reduction
modulo the ideal generated by 
these
axioms as in \refdef{def:reduction-operator}. 
We will show that under the assumptions in
\refth{th:FVStructureTheorem} it holds that 
this operator
satisfies the conditions in
\reflem{lem:PCDegreeLowerBounds}.
Let us
first introduce some notation for
the set of all polynomials that can be generated from some axioms
$\fpartition' \subseteq \fpartition$.

\begin{definition}
  \label{def:I-E-F-ideal}
  For \afvstructure and
  $\fpartition' \subseteq \fpartition$,
  we write  $\spanofc$ to denote the ideal
  generated by the polynomial axioms in $\fpartition' \land
  \fixedconstraints$.  \footnote{That is, 
    $\spanofc$ 
    is the smallest set~$\ideali$ of multilinear polynomials that contains all
    axioms in
    $\fpartition' \land \fixedconstraints$
    and that is closed under addition
    of 
    $\polyp_1, \polyp_2 \in \ideali$
    and by multiplication of
    $\polyp  \in \ideali$
    by any multilinear polynomial over 
    $\vars{\fpartition \land \fixedconstraints}$
    (where as before the resulting product is implicitly
    multilinearized).}
\end{definition}
 
\begin{definition}[    \fvstructure 
    reduction]\label{def:FVGraphOperator}
  For \afvstructure $\fvstructurename$, the 
  \introduceterm{\fvstructure reduction $\fvredophead$}
  on a term $\termt$ is defined as
  $\fvredop{\termt} = \redop[{\spanofc[\support{\termt}]}]{\termt}$.
  For a polynomial $\polyp$, 
  we define $\fvredop{\polyp}$ to be
  the linear
  extension of the operator $\fvredophead$ defined on terms.
\end{definition}

\ifthenelse{\boolean{conferenceversion}}{}{  Looking at \refdef{def:FVGraphOperator}, 
  it is not clear that we are making progress.
  On the one hand, we have defined
  $\fvredophead$
  in terms of standard reduction operators modulo ideals, which is nice
  since there is a well-developed machinery for such operators. On the
  other hand, it is not clear how to actually compute
  using~$\fvredophead$. The problem is that if we look at a polynomial
  $\polyp = \sum_i \termt_i$
  and want to compute
  $\fvredop{\polyp}$,
  then as we expand
  $\fvredop{\polyp} = \sum_i \fvredop{\termt_i}$
  we end up reducing terms in one and the same polynomial modulo
  a~priori completely different ideals.  
  How can we get any sense of what
  $\polyp$ reduces to in such a case?
  The answer is that
  if our
  \fvstructure
  is a good enough expander, then this is not an issue at all.
  Instead, it turns out that we can
  pick a suitably large ideal containing the support of
  all the terms in~$\polyp$ and reduce~$\polyp$ modulo this larger ideal
  instead without changing anything. This key result is proven in
  \reflem{lem:StrongClosureLemma}
  below. 
  To establish this lemma, we need to develop a better understanding
  of polynomial support.
}

\subsection{Some Properties of Polynomial Support}

A crucial technical property that we will need 
is that if  \afvstructure is a good expander in the sense of
\refdef{def:F-V-boundary-expansion},
then
for small enough sets~$\vpartition'$ all
\svcontainedstd 
subsets
$\fpartition' \subseteq \fpartition$ 
as per \refdef{def:support}
are of at most half of the allowed size.

\begin{lemma}\label{lem:SizeOfRelevantSet}
  Let $\fvstructurenot$ be an
  \rboundaryexpandershortstd
  and let
  $\vpartition' \subseteq \vpartition$ 
  be such that
  $\setsize{\vpartition'} \leq \expfactor \expsize / 2 - \expslack$.
  Then it holds that 
  every \svcontainedstd subset 
  $\fpartition' \subseteq \fpartition$ 
  is in fact
  \shalfvcontainedstd.
\end{lemma}

\begin{proof}
  As $\setsize{\fpartition'} \leq \expsize$
  we can 
  appeal to the expansion property of the \fvstructure
  to derive the inequality
  $\setsize{\sboundary{\fpartition'}} \geq \expfactor
  \setsize{\fpartition'} - \expslack$.
  In the other direction, we can obtain an upper bound on the size of
  $\sboundary{\fpartition'}$ by noting that for any 
  \svcontainedstd set
  $\fpartition'$ it holds that
  $\setsize{\sboundary{\fpartition'}} \leq \setsize{\vpartition'}$.
  If we combine these bounds
  and use the assumption that
  $\setsize{\vpartition'} \leq \expfactor \expsize / 2 - \expslack$,
  we can conclude that
  $\setsize{\fpartition'} \leq \expsize / 2$, which proves that 
  $\fpartition'$ 
  is
  \shalfvcontainedstd.
\end{proof}

Even more importantly, 
\reflem{lem:SizeOfRelevantSet}
now allows us to conclude that for a small enough
subset~$\vpartition'$ on the right-hand side of~$\fvstructurenot$
it holds that in fact the whole 
polynomial \ssupport~$\support{\vpartition'}$ 
of $\vpartition'$ on the left-hand side 
is \shalfvcontainedstd.

\begin{lemma}
  \label{lem:MainPropertyOfSupport}
  Let $\fvstructurenot$ be an
  \rboundaryexpandershortstd
  and let
  $\vpartition' \subseteq \vpartition$
  be such that
  $\setsize{\vpartition'} \leq \expfactor \expsize / 2 - \expslack$.
  Then the \ssupport
  $\support{\vpartition'}$ 
  of $\vpartition'$ 
  \wrt $\fvstructurenot$
  is
  \shalfvcontainedstd.
\end{lemma}

\begin{proof}
  We show that for any pair of \svcontainedstd sets
  $\fpartition_1, \fpartition_2 \subseteq \fpartition$
  their union 
  $\fpartition_1 \union \fpartition_2$
  is also \svcontainedstd.
  First, by \reflem{lem:SizeOfRelevantSet} 
  we have
  $\setsize{\fpartition_1}, \setsize{\fpartition_2} \leq \expsize / 2$
  and hence
  $\setsize{\fpartition_1 \union \fpartition_2} \leq
  \expsize$.
  Second, 
  it holds that
  $\sboundary{\fpartition_1}, \sboundary{\fpartition_2} \subseteq
  \vpartition'$,
  which implies that
  $\sboundary{\fpartition_1 \union \fpartition_2} \subseteq
  \vpartition'$,
  because taking the union of two sets can only shrink the
  boundary. 
  This establishes that
  $\fpartition_1 \union \fpartition_2$
  is \svcontainedstd.

  By induction on the number of \svcontainedstd sets we can conclude
  that the support
  $\support{\vpartition'}$ is \svcontainedstd as well, 
  after which one final
  application of \reflem{lem:SizeOfRelevantSet} shows that 
  this set
  is \shalfvcontainedstd. This completes the proof.
\end{proof}

What the next lemma says is, roughly, that 
if we reduce a term~$\termt$ modulo an ideal generated by a 
not too large
set of polynomials containing some polynomials
outside of the
support of~$\termt$, then we can remove
all such polynomials from the generators of the  ideal
without changing the irreducible component of~$\termt$.

\begin{lemma}
  \label{lem:MainLemmaStrongClosure}
  Let $\fvstructurename$ be \afvstructure and let $\termt$ be any
  term.  Suppose that $\fpartition' \subseteq \fpartition$ is such
  that $\fpartition' \supseteq \support{\termt}$ and
  $\setsize{\fpartition'} \leq \expsize$. Then for any term $\termt'$
  with
  $\neighbor{\termt'} \subseteq \neighbor{\support{\termt}} \union
  \neighbor{\termt}$
  it holds that if $\termt'$ is reducible modulo $\spanofc$, it is
  also reducible modulo $\spanofc[\support{\termt}]$.
\end{lemma}

\begin{proof}
  If  $\fpartition'$ is \svcontainedneighbourt, then  by
  \refdef{def:support} it holds that
  $\fpartition' \subseteq \support{\termt}$
  and there is nothing to prove.
  Hence, assume $\fpartition'$ is not \svcontainedneighbourt.
  We claim that this implies that  we can find
  a subformula~$\formf \in \fpartition' \setminus \support{\termt}$ 
  with a neighbouring subset of variables
  \mbox{$ \variablesv \in \bigl( \sboundary{\fpartition'} \intersection
    \neighbor{\formf} \bigr) \setminus \neighbor{\termt'}$}
  in the \rsatrespectful boundary of~$\fpartition'$ but not in the
  neighbourhood of~$\termt'$.
  To argue this, note that since $\setsize{\fpartition'} \leq \expsize$
  it follows from \refdef{def:support} that
  the reason $\fpartition'$ is not \svcontainedneighbourt is that
  there  exist some $\formf \in \fpartition'$
  and some set of variables $\variablesv \in \neighbor{\formf}$ such
  that
  $\variablesv \in \sboundary{\fpartition'} \setminus
  \neighbor{\termt}$.
  Moreover, the assumption $\fpartition' \supseteq \support{\termt}$
  implies that such an~$\formf$ cannot be in $\support{\termt}$.
  Otherwise there would exist an \sntcontained set 
  $\fpartition^*$
  such that
  $
  \formf \in \fpartition^* 
  \subseteq \support{\termt} 
  \subseteq \fpartition'
  $,
  from which it would follow that
  $
  \variablesv   \in 
  \sboundary{\fpartition'} \intersection \neighbor{\fpartition^*}
  \subseteq \sboundary{\fpartition^*}
  \subseteq \neighbor{\termt}
  $,
  contradicting
  $\variablesv \notin \neighbor{\termt}$.
  We have shown that
  $\formf \notin \support{\termt} \subseteq \fpartition'$
  and
  $\variablesv \in \sboundary{\fpartition'} \intersection
  \neighbor{\formf}$,
  and by combining these two facts
  we can also deduce that
  $\variablesv \notin \neighbor{\support{\termt}}$,
  since otherwise $\variablesv$~could not be contained in the
  boundary of~$\fpartition'$. 
  In particular, this means that
  $
  \variablesv \notin
  \neighbor{\termt'} \subseteq \neighbor{\support{\termt}} \union
  \neighbor{\termt}
  $,
  which establishes the claim made above.

  Fixing
  $\formf$ and $\variablesv$ such that
  $\formf \in \fpartition' \setminus \support{\termt}$ and
  $ \variablesv \in \bigl( \sboundary{\fpartition'} \intersection
  \neighbor{\formf} \bigr) \setminus \neighbor{\termt'}$,
  our second claim is that if $\formf$ is removed from the generators
  of the ideal, 
  it still holds that if
  $\termt'$ is reducible modulo~$\spanofc$, then this term is also
  reducible modulo 
  $\spanofc[\fpartition' \setminus \set{\formf}]$.  
  Given this second claim we are done, since
  we can then argue by induction over the elements in
  $\fpartition' \setminus \support{\termt}$ and remove them one by one
  to arrive at the conclusion that every term $\termt'$ with
  $\neighbor{\termt'} \subseteq \neighbor{\support{\termt}} \union
  \neighbor{\termt}$
  that is reducible modulo $\spanofc$ is also reducible modulo
  $\spanofc[\support{\termt}]$, which is precisely what the lemma says.

  We proceed to establish this second claim.
    The assumption that $\termt'$ is reducible modulo
  $\spanofc[\fpartition']$ means that there exists a
  polynomial
  \mbox{$\polyp \in\spanofc$} such that
  $\termt' = \leadingterm{\polyp}$.  Since $\polyp$ is in the
  ideal~$\spanofc$ it can be written as a polynomial combination
  $\polyp = \sum_{i} \polyp_{i} \polyax_{i}$ of axioms
  $\polyax_{i} \in \fpartition' \land \fixedconstraints$ for some
  polynomials $\polyp_{i}$.  If we could hit $\polyp$ with a
  restriction that satisfies (and hence removes)~$\formf$ while
  leaving $\termt'$ and
  $(\fpartition' \setminus \set{\formf}) \land \fixedconstraints$
  untouched, this would show that $\termt'$ is the leading term of
  some polynomial combination of axioms in
  $(\fpartition' \setminus \set{\formf}) \land \fixedconstraints$.
  This is almost what we are going to do.
  
  As our 
  restriction
  $\assignmenta$
  we choose any assignment
  with domain $\domainof{\assignmenta} = \variablesv$
  that \rsatisfies $\formf$. Note that  at least
  one such assignment exists since 
  $\variablesv \in \sboundary{\fpartition'}
  \intersection \neighbor{\formf}$
  is \anrneighbour of~$\formf$
  by \refdef{def:respectful-boundary}. 
  By the choice of $\assignmenta$ it holds that $\formf$ is
  satisfied, \ie that all axioms in $\formf$ are set to~$0$.
  Furthermore, 
  none of the axioms in $\fpartition' \setminus \set{\formf}$ are
  affected by~$\assignmenta$ since $\variablesv$ is in the 
  boundary of~$\fpartition'$.  \footnote{Recalling the remark after \refdef{def:f-v-graph},
    we note that we can ignore here if $\assignmenta$ happens to
    falsify axioms in $\fpartition \setminus \fpartition'$.}
  As for axioms in $\fixedconstraints$ it is not necessarily true that
  $\assignmenta$~will leave all of them untouched, but by assumption
  $\assignmenta$ respects~$\fixedconstraints$ and so any axiom in
  $\fixedconstraints$ is either satisfied (and zeroed out)
  by~$\assignmenta$ or is left  intact. 
  It follows that
  $\restrict{\polyp}{\assignmenta}$ can be be written as a polynomial
  combination 
  $\restrict{\polyp}{\assignmenta}
  = \sum_{i} 
  \bigl( \restrict{\polyp_{i}}{\assignmenta} \bigr) \polyax_{i}$, where
  $\polyax_{i} \in (\fpartition' \setminus \set{\formf}) \land
  \fixedconstraints$, 
  and hence
  $\restrict{\polyp}{\assignmenta} \in 
  \spanofc[\fpartition' \setminus \set{\formf}]$.

  To see that $\termt'$ is preserved as the leading term of
  $\restrict{\polyp}{\assignmenta}$, note that $\assignmenta$ does not
  assign any variables in $\termt'$ since
  $\variablesv \notin \neighbor{\termt'}$.
  Hence, $\termt' = \leadingterm{\restrict{\polyp}{\assignmenta}}$, as
  $\assignmenta$ can only make the other terms smaller \wrt~$\mlt$.
  This shows that there is a polynomial
  $\polyp' = \restrict{\polyp}{\assignmenta} \in \spanofc[\fpartition'
  \setminus \set{\formf}]$
  with $\leadingterm{\polyp'} = \termt'$, and hence $\termt'$ is
  reducible modulo $\spanofc[\fpartition' \setminus \set{\formf}]$.
  The lemma follows.
\end{proof}

We need to deal with one more detail before we can prove the key
technical lemma that it is possible to reduce modulo suitably chosen
\ifthenelse{\boolean{conferenceversion}}{  larger ideals without changing the reduction operator.
  We refer to the upcoming full-length version for the proof
  of the next lemma.}
{  larger ideals without changing the reduction operator,
  namely (again roughly speaking) that reducing a term modulo an ideal
  does not introduce any new variables outside of the generators of
  that ideal.} 

\begin{lemma}
  \label{lem:ContainmentLemma}
  Suppose that
  $\fpartition^* \subseteq \fpartition$ 
  for some \fvstructure 
  and let
  $\termt$
  be any term. Then it holds that
  $\Neighbor{\redop[{\spanofc[\fpartition^*]}]{t}} \subseteq
  \neighbor{\fpartition^*} \union \neighbor{t}$.
\end{lemma}

\ifthenelse{\boolean{conferenceversion}}{}{
\begin{proof}
  Let
  $\polyp = \redop[{\spanofc[\fpartition^*]}]{t}$ 
  be the polynomial obtained when reducing~$\termt$
  modulo~$\spanofc[\fpartition^*]$ 
  and let 
  $\variablesv \in \vpartition$ be any set such that
  $\variablesv \notin \neighbor{\fpartition^*} \union
  \neighbor{\termt}$. 
  We show that 
  $\variablesv \notin \neighbor{\polyp}$.

  By the definition of \fvstructure{}s there exists an
  assignment
  $\assignmenta$ to all of the variables in~$\variablesv$ that
  respects~$\fixedconstraints$.
  Write  
  $\termt = \polyq + \polyp$ 
  with
  $\polyq \in \spanofc[\fpartition^*]$
  and
  $\polyp$ a linear combination of irreducible monomials 
  as in
  \reffact{fact:unique-sum}
  and apply the restriction $\assignmenta$ to
    this equality.
  Note that
  $\restrict{\termt}{\assignmenta} = \termt$
  as
  $\variablesv$ is not a neighbour of $\termt$. Moreover,
  $\restrict{\polyq}{\assignmenta}$ is in the ideal
  $\spanofc[\fpartition^*]$ because $\assignmenta$ does not set any
  variables in $\fpartition^*$ and every axiom in $\fixedconstraints$
  sharing variables with $\variablesv$ is set to $0$ by
  $\assignmenta$. Thus, $\termt$ can be written as
  $\termt = \polyq' + \restrict{\polyp}{\assignmenta}$,
  with
  $\polyq' \in \spanofc[\fpartition^*]$. As all terms in $\polyp$ 
  are
  irreducible modulo $\spanofc[\fpartition^*]$, they remain irreducible
  after restricting $\polyp$ by $\assignmenta$
  by \refobs{obs:restrictions-preserve-irreducibility}. 
  Hence, it follows that
  $\restrict{\polyp}{\assignmenta} = \polyp$ by the uniqueness 
  in  \reffact{fact:unique-sum}
  and $\polyp$ cannot contain any variable from $\variablesv$.
  This in turn implies that every set
  $\variablesv \in \neighbor{\polyp}$ is 
  contained in
  $\neighbor{\fpartition^*} \union \neighbor{\termt}$.
\end{proof}
}

Now we can state the formal claim that enlarging the ideal does not
change the reduction operator if the enlargement is done in the right
way.

\begin{lemma}
  \label{lem:StrongClosureLemma}
  Let $\fvstructurename$ be  \afvstructure 
  and let $\termt$ be any term.
  Suppose that
  $\fpartition' \subseteq \fpartition$
  is such that
  $\fpartition' \supseteq \support{\termt}$ and
  $\setsize{\fpartition'} \leq \expsize$.
  Then  it holds that
  $\redop[\spanofc]{\termt} =
  \redop[{\spanofc[\support{\termt}]}]{\termt}$.
\end{lemma}

\begin{proof}
  We prove that
  $\redop[\spanofc]{\termt} =
  \redop[{\spanofc[\support{\termt}]}]{\termt}$
  by applying the contrapositive of
  \reflem{lem:MainLemmaStrongClosure}.  Recall that this lemma states
  that any term~$\termt'$ with
  $\neighbor{\termt'} \subseteq \neighbor{\support{\termt}} \union
  \neighbor{\termt}$
  that is reducible modulo~$\spanofc$ is also reducible
  modulo~$\spanofc[\support{\termt}]$.  Since every term $\termt'$ in
  $\redop[{\spanofc[\support{\termt}]}]{\termt}$ is irreducible
  modulo~$\spanofc[\support{\termt}]$ and since by
  applying
  \reflem{lem:ContainmentLemma} 
  with $\fpartition^* = \support{\termt}$
  we have that
  $\neighbor{\termt'} \subseteq \neighbor{\support{\termt}} \union
  \neighbor{\termt}$,
  it follows that $\termt'$ is also irreducible modulo $\spanofc$.
  This shows that
  $\redop[\spanofc]{\termt} =
  \redop[{\spanofc[\support{\termt}]}]{\termt}$
  as claimed, and the lemma follows.
\end{proof}

\subsection{Putting the Pieces in the Proof Together}

\ifthenelse{\boolean{conferenceversion}}           
{We just need two more lemmas to establish
  \refth{th:FVStructureTheorem}.
  To keep the length of this extended abstract reasonable, we just
  state these lemmas and hint at how to prove them.}
{Now we have just a couple of lemmas left before we can prove
  \refth{th:FVStructureTheorem},
  which as discussed above will be  established
  by appealing to
  \reflem{lem:PCDegreeLowerBounds}.}

\begin{lemma}\label{lem:SizeOfSupportFromDegree}
  Let 
  $\fvstructurenot$ 
  be an \rboundaryexpandershortstd{} with 
  \voverlapterm $\voverlapnot{\vpartition} = \setdegd$. 
  Then for any term
  $\termt$ with
  $\mdegreeof{\termt} \leq (\expfactor \expsize - 2 \expslack) / (2
  \setdegd)$
  it holds that $\setsize{\support{\termt}} \leq \expsize / 2$.
\end{lemma}

\ifthenelse{\boolean{conferenceversion}}{  This is a fairly straightforward application of
  \reflem{lem:MainPropertyOfSupport}. 
}{\begin{proof}
  Because of the bound on the \voverlapterm
  $\voverlapnot{\vpartition}$ 
  we have that the size of
  $\neighbor{\termt}$ is bounded by
  $\expfactor \expsize / 2 - \expslack$.  An application of
  \reflem{lem:MainPropertyOfSupport} now yields the desired bound
  $\setsize{\support{\termt}} \leq \expsize / 2$.
\end{proof}
}

\begin{lemma}\label{lem:SupportsOfReductionOperator}
  Let 
  $\fvstructurenot$ 
  be an \rboundaryexpandershortstd{} with 
  \voverlapterm $\voverlapnot{\vpartition} = \setdegd$. 
  Then for any term
  $\termt$ with
  $\mdegreeof{\termt} < \floor{(\expfactor \expsize - 2 \expslack) /
    (2 \setdegd)}$,
  any term $\termt'$ occurring in
  $\redop[{\spanofc[\support{\termt}]}]{\termt}$, and any variable
  $\varx$, it holds that
  $\redop[{\spanofc[\support{\varx \termt'}]}]{\varx \termt'} =
  \redop[{\spanofc[\support{\varx \termt}]}]{\varx \termt'}$.
\end{lemma}

\ifthenelse{\boolean{conferenceversion}}{  This lemma follows from \refobs{obs:monotone-support},
  \reflem{lem:ContainmentLemma}, \reflem{lem:StrongClosureLemma},
  and \reflem{lem:SizeOfSupportFromDegree}.
}{\begin{proof}
  We prove the lemma by showing that
  $\support{\varx \termt'} \subseteq \support{\varx \termt}$ 
  and that
  $\setsize{\support{\varx \termt}} \leq \expsize$, 
  which then allows
  us to apply \reflem{lem:StrongClosureLemma}. To prove 
  that
  $\support{\varx \termt'}$ is a subset of~$\support{\varx \termt}$,
  we 
  will   
  show that
  $\support{\varx \termt'} \union \support{\varx \termt}$ is
  \svcontained{\expsize}{\neighbor{\varx \termt}}
  in the sense of
  \refdef{def:support}.
  From this it follows that
  $
  \support{\varx \termt'}
  \subseteq
  \support{\varx \termt'} \union \support{\varx \termt}
  =
  \support{\varx \termt}
  $.

  Towards this goal,
  as $\mdegreeof{\termt'} \leq \mdegreeof{\termt}$
  we first observe that 
  we can apply \reflem{lem:SizeOfSupportFromDegree} to
  deduce that
  $\setsize{\support{\varx \termt'}}   \leq \expsize / 2$
  and 
  \mbox{$\setsize{\support{\varx \termt}} \leq \expsize / 2$},
  and hence
  $\setsize{\support{\varx \termt'} \union \support{\varx \termt}}
  \leq \expsize$,
  which satisfies the size condition for containment.
  It remains to show that
  $\Sboundary{\support{\varx \termt'} \union \support{\varx \termt}}
  \subseteq \neighbor{\varx \termt}$.
  From \reflem{lem:ContainmentLemma} we have that
  $\neighbor{\termt'} \subseteq \neighbor{\support{\termt}} \union
  \neighbor{\termt}$.
  As
  $\neighbor{\varx \termt'} = \neighbor{\varx} \union
  \neighbor{\termt'}$
  and $\support{\termt} \subseteq \support{\varx \termt}$
  by the monotonicity in \refobs{obs:monotone-support},
  it 
  follows that
  \begin{equation}
    \label{eq:support-multiplication}
    \neighbor{\varx \termt'} 
    =
    \neighbor{\varx} \union \neighbor{\termt'} 
    \subseteq
    \neighbor{\varx} \union
    \neighbor{\support{\termt}} \union
    \neighbor{\termt} 
    \subseteq
    \neighbor{\support{\varx \termt}} \union
    \neighbor{\varx \termt} 
    \eqperiod
\end{equation}
  If we now consider the \rboundary of the set
  $\support{\varx \termt'} \union \support{\varx \termt}$,
  it holds that
  \begin{equation}
    \begin{split}
    \SBOUNDARY{\support{\varx \termt'} 
      \union \support{\varx \termt}}  = 
        \hspace{-3.0cm}&
    \\
    &= 
    \left(\SBOUNDARY{\support{\varx \termt'}} \setminus
    \NEIGHBOR{\support{\varx \termt}}\right) 
    \union
    \left(\SBOUNDARY{\support{\varx \termt}} \setminus
    \NEIGHBOR{\support{\varx \termt'}}\right) 
    \\
    &\subseteq 
    \left(\NEIGHBOR{\varx \termt'} \setminus 
    \NEIGHBOR{\support{\varx \termt}}\right) 
    \union
    \left(\NEIGHBOR{\varx \termt} \setminus 
    \NEIGHBOR{\support{\varx \termt'}}\right) 
    \\
    &\subseteq \NEIGHBOR{\varx \termt}
    \eqcomma
    \end{split}
  \end{equation}
  where the first line follows from 
  the boundary definition in   \refdef{def:respectful-boundary},
  the second line follows by the property of 
  \ssupport that
  $\sboundary{\support{\varx \termt}} \subseteq \neighbor{\varx
    \termt}$,
  and the last line follows 
  from~\refeq{eq:support-multiplication}.
  Hence, $\support{\varx \termt'} \union \support{\varx \termt}$ is
  \svcontained{\expsize}{\neighbor{\varx \termt}}.

  As discussed above, we can now apply
  \reflem{lem:StrongClosureLemma} to reach the desired conclusion that 
  the equality
  $\redop[{\spanofc[\support{\varx \termt'}]}]{\varx \termt'} =
  \redop[{\spanofc[\support{\varx \termt}]}]{\varx \termt'}$
  holds.
\end{proof}
}

\ifthenelse{\boolean{conferenceversion}}
{}
{Now we can prove our main technical theorem.}

\begin{proof}[Proof of \refth{th:FVStructureTheorem}]
  Recall that the assumptions of the theorem are that we have
  \afvstructure for a CNF formula 
  $\originalcnf = 
  \Land_{\formf \in \fpartition} \formf \land \fixedconstraints$
  such that
  $\fvstructurenot$ is an 
  \rboundaryexpandershortstd{} with  
  \voverlapterm   $ \voverlapnot{\vpartition} = \setdegd $
  and that furthermore 
  for all
  $\fpartition' \subseteq \fpartition$, 
  \mbox{$\setsize{\fpartition'} \leq \expsize$},
  it holds that 
  $\Land_{\formf \in \fpartition'} \formf \land \fixedconstraints$
  is satisfiable. 
  We want to prove that no polynomial calculus derivation from
  $\Land_{\formf \in \fpartition} \formf \land \fixedconstraints
  =
  \fpartition \land \fixedconstraints$
  of degree at most
  $(\expfactor \expsize - 2 \expslack) / (2 \setdegd)$
  can reach contradiction. 

  \ifthenelse{\boolean{conferenceversion}}{}{   First, if removing all axiom clauses from
  $
  \fpartition \land \fixedconstraints
  $
  with
  degree strictly greater than 
  \mbox{$(\expfactor \expsize - 2 \expslack) / (2 \setdegd)$} 
  produces a satisfiable formula, then the lower bound trivially
  holds. Otherwise, we can remove these large-degree axioms 
  and still be left with
  \afvstructure that satisfies the conditions 
  above.
  In order to see this,
  let us analyze what happens to the \fvstructure if an axiom is
  removed from the formula. 

  Removing axioms from $\fixedconstraints$ only relaxes the conditions
  on \rsatrespectful satisfiability while keeping all 
  edges in the graph, so the
  conditions of the theorem still hold. In removing axioms from
  $\fpartition$ we have two cases: either we remove all axioms from some
  subformula
  $\formf \in \fpartition$ or we remove only a part of this
  subformula. 
    In the former case, it is clear that
  we can remove the vertex~$\formf$ from the structure
  without affecting any of the conditions. 
    In the latter case, 
    we claim
  that any set $\variablesv \in \vpartition$ that
    is
  an \rneighbour
  of~$\formf$ remains an \rneighbour of the formula~$\formf'$ in which 
  large degree axioms
    have been
  removed. 
    Clearly,
  the same assignments to $\variablesv$ that satisfy $\formf$ also satisfy
  $\formf' \subseteq \formf$.
          Also, $\variablesv$ must still be a neighbour of~$\formf'$, for otherwise
  $\formf'$ would not share any variables with $\variablesv$, which would
  imply that no assignment to $\variablesv$ could satisfy $\formf'$ and
  hence $\formf$. This would contradict the assumption that
  $\variablesv$ is an \rneighbour of $\formf$. 
  Hence, we conclude that removal of
    large-degree axioms can only improve the \rboundary expansion
  of the \fvstructure.
}
 
\ifthenelse{\boolean{conferenceversion}}{  We can focus on \afvstructure where the degree of axioms in
  $\fpartition \land \fixedconstraints$ is at most
  $(\expfactor \expsize - 2 \expslack) / (2 \setdegd)$, as 
  it is not hard to show that 
  axioms of higher degree can safely be ignored.}{    Thus, let us focus on \afvstructure $\fvstructurename$ that has
  all axioms of degree at most
  $(\expfactor \expsize - 2 \expslack) / (2 \setdegd)$.}   We want to show that the operator $\fvredophead$ from
  \refdef{def:FVGraphOperator} satisfies the conditions of
  \reflem{lem:PCDegreeLowerBounds}, from which
  \refth{th:FVStructureTheorem}
  immediately follows.
  We can note right away that the operator~$\fvredophead$ is linear by
  construction. 

  To prove that
  $\fvredop{1} = \redop[{\spanofc[\support{1}]}]{1} \neq 0$, 
  we start by
  observing
  that the size of the \ssupport of~$1$ is upper-bounded by
  $\expsize / 2$
  according to \reflem{lem:SizeOfSupportFromDegree}. 
  Using the assumption that for
  every subset $\fpartition'$ of $\fpartition$,
  $\setsize{\fpartition'} \leq \expsize$, the formula
  $\fpartition' \land \fixedconstraints$ is satisfiable, it follows
  that~$1$ is not in the ideal $\spanofc[\support{1}]$ and hence
  $\redop[{\spanofc[\support{1}]}]{1} \neq 0$.

  We next show that 
  $\fvredop{\polyax} = 0$
  for any axiom clause
  $\polyax \in \fpartition \land \fixedconstraints$
  (where we recall that we identify a clause $\polyax$ with its
  translation into a linear combination of monomials).
  By the   \ifthenelse{\boolean{conferenceversion}}{assumption}  {preprocessing step}   above it holds that
  the degree of $\polyax$ is bounded by
  $(\expfactor \expsize - 2 \expslack) / (2 \setdegd)$,  
  from which it follows by \reflem{lem:SizeOfSupportFromDegree} that
  the size of the
  \ssupport of every term in $\polyax$ is bounded by~$\expsize / 2$.
  Since $\polyax$ is the polynomial encoding of a clause, 
  the leading term~$\leadingterm{\polyax}$ 
  contains all the variables appearing in $\polyax$.  \footnote{We remark that this is the only place in the proof where
    we are using that $\polyax$ is (the encoding of) a clause.}
 Hence,
  the \ssupport $\support{\leadingterm{\polyax}}$ 
  of the leading term contains the \ssupport of every
  other term in~$\polyax$ by \refobs{obs:monotone-support} and we can
  use
  \reflem{lem:StrongClosureLemma} to conclude that
  $\fvredop{\polyax} = \redop[{\spanofc[\support{\leadingterm{\polyax}}]}]{\polyax}$.
  If $\polyax \in \fixedconstraints$,
  this means we are done
  because
  $\spanofc[\support{\leadingterm{\polyax}}]$ contains 
  all of $\fixedconstraints$,
  implying that $\fvredop{\polyax} = 0$.

  For $\polyax \in \fpartition$ we cannot
  immediately argue that $\polyax$ reduces to~$0$, 
  since (in contrast to~\cite{AR03LowerBounds})
  it is not immediately clear that
  $\support{\leadingterm{\polyax}}$ contains $\polyax$.
  The problem here is that we might worry that $\polyax$ is part of
  some subformula 
  $\formf \in \fpartition$ 
  for which the boundary
  $\sboundary{\formf}$
  is not contained in
  $\neighbor{\leadingterm{\polyax}} = \vars{\polyax}$, 
  and hence there is no obvious reason why
  $\polyax$ should be a member of any
  \svcontained{\expsize}{\neighbor{\leadingterm{\polyax}}}
  subset of~$\fpartition$.
  However, in view of \reflem{lem:StrongClosureLemma}
  (applied, strictly speaking, once for every term in~$\polyax$)
  we can choose some 
  $\formf \in \fpartition$ such that $\polyax \in \formf$ and add it
  to the \ssupport 
  $\support{\leadingterm{\polyax}}$ to obtain a set 
  $\fpartition' = \support{\leadingterm{\polyax}} \union \set{\formf}$
  of size 
  $\setsize{\fpartition'} 
  \leq \expsize / 2 + 1
  \leq \expsize$
  such that
  $
  \redop[{\spanofc[\support{\leadingterm{\polyax}}]}]{\polyax}
  =
  \redop[\spanofc]{\polyax}
  $.
  Since $\spanofc$ contains $\polyax$ as a generator we conclude that
  $\fvredop{\polyax} = \redop[\spanofc]{\polyax} = 0$ also 
  \mbox{for $\polyax \in \fpartition$.}  \footnote{Actually, a slighly more careful argument reveals that
    $\polyax$ is always contained in
    $\support{\leadingterm{\polyax}}$.
    This is so since for any $\formf \in \fpartition$ with $\polyax
    \in \formf$  it holds that any neighbours in
    $\neighbor{\formf} \setminus \neighbor{\leadingterm{\polyax}}$
    have to be \rsatdisrespectful, and so 
    such an~$\formf$ always makes it into the support.
    However, the reasoning gets a bit more involved, and since we
    already needed to use \reflem{lem:StrongClosureLemma} anyway we
    might as well apply it once more here.}

  It remains to prove
  the last property in
  \reflem{lem:PCDegreeLowerBounds} stating that
  $\fvredop{\varx \termt} = \fvredop{\varx \fvredop{\termt}}$ 
    for any term~$\termt$ such that
  $\mdegreeof{\termt} < \floor{(\expfactor \expsize - 2 \expslack) /
    (2 \setdegd)}$. 
  We can see that this holds by studying the following sequence of
  equalities: 
  \begin{subequations}   
  \begin{align*}
    \fvredop{\varx \fvredop{\termt}} 
    &= \sum_{\termt' \in \fvredop{\termt}} \fvredop{\varx \termt'}
    & & \bigl[\text{by linearity}\bigr] \\
    &= \sum_{\termt' \in \fvredop{\termt}} \redop[{\spanofc[\support{\varx \termt'}]}]{\varx \termt'}
    & & \bigl[\text{by definition of $\fvredophead$}\bigr] \\
    &= \sum_{\termt' \in \fvredop{\termt}} \redop[{\spanofc[\support{\varx \termt}]}]{\varx \termt'}
    & & \bigl[\text{by \reflem{lem:SupportsOfReductionOperator}}\bigr] \\
    &= \redop[{\spanofc[\support{\varx \termt}]}]{\varx \fvredop{\termt}}
    & & \bigl[\text{by linearity}\bigr] \\
    &= \redop[{\spanofc[\support{\varx \termt}]}]{\varx 
      \redop[{\spanofc[\support{\termt}]}]{\termt}}
    & & \bigl[\text{by definition of $\fvredophead$}\bigr] \\
    &= \redop[{\spanofc[\support{\varx \termt}]}]{\varx \termt}
    & & \bigl[\text{by \refobs{obs:reduction-mod-larger-ideal}}\bigr] \\
    &= \fvredop{\varx \termt}
    & & \bigl[\text{by definition of $\fvredophead$}\bigr] 
  \end{align*}
  \end{subequations}
  Thus, $\fvredophead$ satisfies all the properties of
  \reflem{lem:PCDegreeLowerBounds}, 
  from which the theorem follows.
\end{proof}

\ifthenelse{\boolean{conferenceversion}}{}{  Let us next show that if the slack~$\expslack$ in
  \refth{th:FVStructureTheorem} is
    zero,
  then
  the condition that
  $\fpartition' \land \fixedconstraints$ is satisfiable for sufficiently
  small $\fpartition'$
    is already implied by the expansion.}

\ifthenelse{\boolean{conferenceversion}}{}{\begin{lemma}
  \label{lem:SatisfiableSubformulaFromExpansion}
  If \afvstructure is an
  \rboundaryexpandershort{\expsize}{\expfactor}{0}{\fixedconstraints}
  and
  $\vars{\fpartition \land \fixedconstraints} = 
  \Union_{\variablesv \in \vpartition} \variablesv$,
  then for any
  $\fpartition' \subseteq \fpartition$, 
  $\setsize{\fpartition'} \leq \expsize$,
  the formula $\fpartition' \land \fixedconstraints$ is satisfiable.
\end{lemma}
}

\ifthenelse{\boolean{conferenceversion}}{}{\begin{proof}
  Let $\fpartition' \subseteq \fpartition$ be any subset of size at
  most $\expsize$. First, we show that we can find 
  a subset
  $\vpartition' \subseteq \neighbor{\fpartition'}$
  and an assignment~$\assignmenta$ to the set of variables 
  $\Union_{\variablesv \in \vpartition'} \variablesv$
  such that $\assignmenta$ \rsatisfies $\fpartition'$.
  We do this by
  induction on the number of formulas in $\fpartition'$. As the
  \fvstructure is an
  \rboundaryexpandershort{\expsize}{\expfactor}{0}{\fixedconstraints}
  it follows that
  $\setsize{\sboundary{\fpartition'}} \geq \expfactor
  \setsize{\fpartition'} > 0$
  for any non-empty subset $\fpartition'$ and hence there exists a formula
  $\formf \in \fpartition'$ and a variable set $\variablesv'$ such
  that $\variablesv'$ is an \rneighbour of $\formf$ and is not a
  neighbour of any formula in $\fpartition' \setminus \set{\formf}$.
  Therefore, there is an assignment $\assignmenta$ to the variables in
  $\variablesv'$ that \rsatisfies $\formf$. By the induction
  hypothesis there also exists an assignment $\assignmenta'$ that
  \rsatisfies $\fpartition' \setminus \set{\formf}$ and does not
  assign any
  variables in~$\variablesv'$ as
  $\variablesv' \notin \neighbor{ \fpartition' \setminus
    \set{\formf}}$.
  Hence, by extending the assignment $\assignmenta'$ to the variables
  in~$\variablesv'$ according to the assignment~$\assignmenta$, we
  create an assignment to the union of variables in some subset of
  $\neighbor{\fpartition'}$ that \rsatisfies $\fpartition'$.

  We now need to show how to extend this to an assignment satisfying
  also~$\fixedconstraints$.
  To this end, let
  $\assignmenta_{\fpartition'}$ be an assignment that \rsatisfies
  $\fpartition'$ and assigns the variables in
  $\Union_{\variablesv \in \vpartition'} \variablesv$ 
  for some
  $\vpartition' \subseteq \neighbor{\fpartition'}$. By another
  induction over the 
  size
  $\setsize{\vpartition'' \setminus \vpartition'}$ of families
  $\vpartition'' \supseteq \vpartition'$, 
  we show that there is an
  assignment~$\assignmenta_{\vpartition''}$ to the variables
  $\Union_{\variablesv \in \vpartition''} \variablesv$ that
  \rsatisfies $\fpartition'$ for every
  $\vpartition' \subseteq \vpartition'' \subseteq \vpartition$. 
  When
  $\vpartition'' = \vpartition'$, we just take the assignment
  $\assignmenta_{\fpartition'}$. 
  We want to show that 
  for any
  $\variablesv' \in \vpartition \setminus \vpartition''$
  we can
  extend $\assignmenta_{\vpartition''}$ to the variables
  in~$\variablesv'$ so that the new assignment
  \rsatisfies~$\fpartition'$. As $\variablesv'$ \rsatrespects
  $\fixedconstraints$, 
  there is an assignment $\assignmenta_{\variablesv'}$ to the
  variables~$\variablesv'$ that satisfies all affected clauses
  in~$\fixedconstraints$.
  We would like to combine
  $\assignmenta_{\variablesv'}$ 
  and
  $\assignmenta_{\vpartition''}$ 
  into one assignment, but this requires some care since the
  intersection of the domains
  $
  \variablesv' \intersection
  \bigl( \Union_{\variablesv \in \vpartition''} \variablesv \bigr)
  $
  could be non-empty.
  Consider 
  therefore
  the subassignment 
  $\assignmenta_{\variablesv'}^*$ 
  of~$\assignmenta_{\variablesv'}$ 
  that assigns only the variables in
  $\variablesv' \setminus 
  \bigl( \Union_{\variablesv \in \vpartition''} \variablesv \bigr)$.
  We claim that extending 
  $\assignmenta_{\vpartition''}$ by 
  $\assignmenta_{\variablesv'}^*$ creates an assignment that
  \rsatrespects~$\fixedconstraints$. This is because every clause in
  $\fixedconstraints$ that has a variable in $\variablesv'$ and was
  not already satisfied by~$\assignmenta_{\vpartition''}$ cannot have
  variables in
  $\variablesv' \intersection \bigl( \Union_{\variablesv \in
    \vpartition''} \variablesv \bigr)$
  (if so, $\assignmenta_{\vpartition''}$ would have been \rnegadjective)
  and hence 
  every such clause
  must be satisfied by the subassignment~  $\assignmenta_{\variablesv'}^*$.

  Thus, we can find an assignment to all the variables
  $\union_{\variablesv \in \vpartition} \variablesv$ that \rsatisfies
  $\fpartition'$. As $\vpartition$ includes all the variables in
  $\fixedconstraints$ it means that $\fixedconstraints$ is also fully
  satisfied. Hence, $\fpartition' \land \fixedconstraints$ is
  satisfiable and the lemma follows.
\end{proof}
}

\ifthenelse{\boolean{conferenceversion}}
{We conclude the section by stating the following
  version of  \refth{th:FVStructureTheorem} for the most commonly
  occuring case with standard expansion without any slack.}
{This allows us to conclude this section by stating the following
  version of  \refth{th:FVStructureTheorem} for the most commonly
  occuring case with standard expansion without any slack.}

\begin{corollary}
  \label{cor:FVStructureTheoremWithoutSlack}
  Suppose that $\fvstructurenot$ is an
  \rboundaryexpandershort{\expsize}{\expfactor}{0}{\fixedconstraints}
  with
  \voverlapterm   $ \voverlapnot{\vpartition} = \setdegd $
  such that 
  $\vars{\fpartition \land \fixedconstraints} = 
  \Union_{\variablesv \in \vpartition} \variablesv$.
  Then any polynomial calculus refutation of the formula
  $\Land_{\formf \in \fpartition} \formf \land \fixedconstraints$
  requires degree strictly greater than
  $\expfactor \expsize / (2 \setdegd)$.
\end{corollary}

\ifthenelse{\boolean{conferenceversion}}
{
\begin{proof}[Proof sketch]
  It is not hard to show that 
  if \afvstructure is an
  \rboundaryexpandershort{\expsize}{\expfactor}{0}{\fixedconstraints}
  such that
  $\vars{\fpartition \land \fixedconstraints} = 
  \Union_{\variablesv \in \vpartition} \variablesv$,
  then for any
  $\fpartition' \subseteq \fpartition$, 
  $\setsize{\fpartition'} \leq \expsize$,
  it holds that
  the formula $\fpartition' \land \fixedconstraints$ is satisfiable.
  Now the corollary follows immediately from
  \refth{th:FVStructureTheorem}.
\end{proof}
}
{
\begin{proof}
  This follows immediately by plugging
  \reflem{lem:SatisfiableSubformulaFromExpansion} into
  \refth{th:FVStructureTheorem}.
\end{proof}
}

\makeatletter{}\section{Applications}
\label{sec:applications}

In this section, we demonstrate how to use the machinery developed in
\refsec{sec:method} to establish degree lower bounds for polynomial
calculus. 
\ifthenelse{\boolean{conferenceversion}}
{As a warm-up, let us consider the bound from~\cite{AR03LowerBounds}
  for CNF formulas $\originalcnf$ whose clause-variable incidence
  graph~$G(\originalcnf)$ are good enough expanders in the following sense.}
{Let us warm up by reproving the bound from~\cite{AR03LowerBounds}
  for CNF formulas $\originalcnf$ whose clause-variable incidence
  graphs~$G(\originalcnf)$ are good enough expanders.  
  We first recall the expansion concept used in~\cite{AR03LowerBounds}
  for ordinary bipartite graphs.}  

\begin{definition}[Bipartite boundary expander]
  \label{def:bipartite-boundary-expander}
  A bipartite graph
  $\graphg = (\vertexsetu \disunion \vertexsetv, \edgesete)$
  is a
  \introduceterm{bipartite \mbox{$(\expsize, \expfactor)$-boundary} expander} if for
  every set of vertices
  $\vertexsetu' \subseteq \vertexsetu, \setsize{\vertexsetu'} \leq s$,
  it holds that
  $\setsize{\boundary{\vertexsetu'}} \geq \expfactor
  \setsize{\vertexsetu'}$,
  where the \introduceterm{boundary}
  $
  \boundary{\vertexsetu'}
  =
  \Setdescr[:]
  {\vertexv \in \vertexsetv}
  {\setsize{\neighbor{\vertexv} \intersection \vertexsetu'} = 1}
  $
  consists of all vertices on the right-hand side~$\vertexsetv$ that
  have a unique neighbour in~$\vertexsetu'$ on the left-hand side.
\end{definition}

We can simply identify the \fvstructure with the standard
clause-variable incidence 
\ifthenelse{\boolean{conferenceversion}}
{graph~$G(\originalcnf)$ (setting $\fixedconstraints = \emptyset$)}
{graph~$G(\originalcnf)$} 
to recover the degree lower bound in~\cite{AR03LowerBounds} as stated
next.

\begin{theorem}[\cite{AR03LowerBounds}]
  For any CNF formula $\originalcnf$ and any constant $\expfactor > 0$
  it holds that
  if the clause-variable incidence graph
  $\graphg(\originalcnf)$ is an $(\expsize,
  \expfactor)$\nobreakdash-boundary expander,  then 
  the polynomial calculus degree required to refute~$\originalcnf$
  in polynomial calculus 
  is 
  $\mdegreeref{\originalcnf} > \expfactor \expsize / 2$.
\end{theorem}

\ifthenelse{\boolean{conferenceversion}}{}{\begin{proof}
  To choose $\graphg(\originalcnf)$ as our \fvstructure, we set
  $\fixedconstraints$ to be the empty formula, 
  $\fpartition$ to be the
  set of clauses of $\originalcnf$ interpreted as one-clause CNF
  formulas, and $\vpartition$ to be the set of variables partitioned
  into singleton sets.
  As $\fixedconstraints$ is an empty formula every set
  $\variablesv$ respects it. Also, every neighbour of some
  clause~$\clc \in \fpartition$ is an \rneighbour because we can set
  the neighbouring variable so that the
  clause~$\clc \in \fpartition$ is satisfied. 
  Under this interpretation   $\graphg(\originalcnf)$ is an
  \rboundaryexpandershort{\expsize}{\expfactor}{0}{\fixedconstraints},
  and hence by
  \refcor{cor:FVStructureTheoremWithoutSlack} the degree of refuting
  $\originalcnf$ is greater than $\expfactor \expsize / 2$.
\end{proof}
}

As a second application, which is more interesting in the sense that
the  \fvstructure is nontrivial, we show how 
the degree lower bound for the ordering principle formulas
in~\cite{GL10Optimality} can be established using this framework.
For an undirected (and in general non-bipartite) graph $\graphg$, 
the \introduceterm{graph ordering principle formula~$\gop$} 
says that there exists a totally ordered set of
$\setsize{\vertices{\graphg}}$~elements
where no element is minimal, since every element/vertex~$\vertexv$ has a
neighbour $\vertexu \in \neighbor{\vertexv}$ which is smaller according
to the ordering.
Formally, the CNF formula~$\gop$
is defined 
over variables~$x_{\vertexu, \vertexv}$,
$\vertexu, \vertexv \in \vertices{\graphg}$, $\vertexu \neq \vertexv$,
where the intended meaning of the variables is that
$x_{\vertexu, \vertexv}$ is true if
$\vertexu < \vertexv$ according to the ordering,
and consists of the  following axiom clauses: 
\begin{subequations}
  \begin{align}
  \label{eq:gop-transitivity}
  & \olnot{x}_{\vertexu, \vertexv} \lor \olnot{x}_{\vertexv, \vertexw} 
  \lor x_{\vertexu, \vertexw} 
  & & 
  \text{$\vertexu, \vertexv, \vertexw \in \variablesv(\graphg), 
    \vertexu \neq \vertexv \neq \vertexw \neq \vertexu$} 
  & & 
  \text{(transitivity)} 
  \\
  \label{eq:gop-antisymmetry}
  & \olnot{x}_{\vertexu, \vertexv} \lor \olnot{x}_{\vertexv, \vertexu} 
  & & 
  \text{$\vertexu, \vertexv \in \variablesv(\graphg), 
    \vertexu \neq \vertexv$} 
  & & 
  \text{(anti-symmetry)} 
  \\
  \label{eq:gop-totality}
  & x_{\vertexu, \vertexv} \lor x_{\vertexv, \vertexu} 
  & & 
  \text{$\vertexu, \vertexv \in \variablesv(\graphg), 
    \vertexu \neq \vertexv$}
  & & 
  \text{(totality)} 
  \\
  \label{eq:gop-nonminimality}
  & \Lor_{\vertexu \in \neighbor{\vertexv}} x_{\vertexu, \vertexv} 
  & & 
  \text{$\vertexv \in V(G)$}
  & &
  \text{(non-minimality)} 
\end{align}
\end{subequations}

We remark that the graph ordering principle on the complete graph
$K_n$ on $n$~vertices is the 
\introduceterm{(linear) ordering principle formula~$\lop[n]$}
(also known as a
\introduceterm{least number principle formula}, or 
\introduceterm{graph tautology} in the literature),
for which the non-minimality axioms~\refeq{eq:gop-nonminimality} have
width linear in~$n$. By instead considering graph ordering formulas
for graphs~$\graphg$ of bounded 
degree, one can bring the initial
width of the formulas down 
so that the question of degree
lower bounds becomes meaningful.

To prove degree lower bounds for~$\gop$ we need the following
extension of boundary expansion to the case of non-bipartite graphs.

\begin{definition}  [Non-bipartite boundary expander]
  \label{def:NonBipartiteBoundaryExpander}
  A graph $\graphg = (\vertexsetv, \edgesete)$ is an
  \introduceterm{\mbox{$(\expsize, \expfactor)$-boundary} expander} 
  if for every subset of vertices
  $\vertexsetv' \subseteq \vertexsetv(\graphg)$,
  $\setsize{\vertexsetv'} \leq \expsize$,
  it holds that
  $\setsize{\boundary{\vertexsetv'}} \geq \expfactor
  \setsize{\vertexsetv'}$,
  where the \introduceterm{boundary}
  $
  \boundary{\vertexsetv'}
  =
  \Setdescr[:]
  {\vertexv \in \vertices{\graphg} \setminus \vertexsetv'}
  {\Setsize{\neighbor{\vertexv} \intersection \vertexsetv'} = 1}
  $
  is the set of all vertices in
  $\vertexsetv(\graphg) \setminus \vertexsetv'$ that have a unique
  neighbour in $\vertexsetv'$.
\end{definition}

We want to point out that the definition of expansion used by Galesi
and Lauria in~\cite{GL10Optimality} is slightly weaker in that 
they do not
require boundary expansion but just vertex expansion
(measured as 
\mbox{$\setsize{\neighbor{V'} \setminus V'}$}
for vertex sets $\vertexsetv'$ with
$\setsize{\vertexsetv'} \leq \expsize$), 
and hence their result is slightly stronger than what we state below
in \refth{th:Galesi-Lauria-theorem}. With some modifications of the
definition of \rboundary in \fvstructure{}s it would be possible to
match the lower bound in~\cite{GL10Optimality}, 
but it would also make the definitions more cumbersome and so we
choose not to do so here.

\begin{theorem}[\cite{GL10Optimality}]
  \label{th:Galesi-Lauria-theorem}
  For a non-bipartite graph $\graphg$ that is an
  $(\expsize, \expfactor)$-boundary expander it holds that 
    $\mdegreeref{\gop} > \expfactor \expsize / 4$.  
\end{theorem}

\ifthenelse{\boolean{conferenceversion}}{  \begin{proof}[Proof sketch]
    To form the \fvstructure for~$\gop$, we let $\fixedconstraints$
    consist of all transitivity axioms~\eqref{eq:gop-transitivity},
    anti-symmetry axioms~\eqref{eq:gop-antisymmetry}, and totality
    axioms~\eqref{eq:gop-totality}.  The non-minimality
    axioms~\eqref{eq:gop-nonminimality} viewed as singleton sets form
    the family $\fpartition$, while $\vpartition$ is the family of
    variable sets $\variablesv_{\vertexv}$ for each vertex $\vertexv$
    containing all variables that mention $\vertexv$, \ie
    $\variablesv_{\vertexv} = \setdescr {x_{\vertexu, \vertexw}}
    {\vertexu, \vertexw \in \vertices{\graphg},\, \vertexu = \vertexv
      \text{ or } \vertexw = \vertexv}$.
    We leave it to the reader to verify that $\fvstructurenot$ is
    an~\rboundaryexpandershort{\expsize}{\expfactor}{0}{\fixedconstraints}
    and that the \voverlapterm $\voverlapnot{\vpartition}$ is $2$,
    which implies the lower bound.
  \end{proof}
}{\begin{proof}
  To form the \fvstructure for~$\gop$, we let $\fixedconstraints$
  consist of all transitivity axioms~\eqref{eq:gop-transitivity},
  anti-symmetry axioms~\eqref{eq:gop-antisymmetry},
  and totality axioms~\eqref{eq:gop-totality}.
  The non-minimality
  axioms~\eqref{eq:gop-nonminimality} viewed as singleton sets 
  form the family $\fpartition$, while $\vpartition$ is the family
  of variable sets $\variablesv_{\vertexv}$ for each vertex
  $\vertexv$ containing all variables that mention
  $\vertexv$, \ie
  $\variablesv_{\vertexv} = 
  \setdescr
  {x_{\vertexu, \vertexw}}
  {\vertexu, \vertexw \in \vertices{\graphg},\, 
    \vertexu = \vertexv \text{ or }
    \vertexw = \vertexv}$.

  For a vertex $\vertexu$, the neighbours of
  a non-minimality axiom
  $\formf_{\vertexu} 
  = \Lor_{\vertexv \in \neighbor{\vertexu}} x_{\vertexv, \vertexu} 
  \in \fpartition$
  are
  variable sets~$\variablesv_{\vertexv}$ where~$\vertexv$ is either
  equal to $\vertexu$ or a neighbour of~$\vertexu$ in~$\graphg$. We
  can prove that each
  $\variablesv_{\vertexv} \in \neighbor{\formf_{\vertexu}}$ 
  is an \rneighbour of~$\formf_{\vertexu}$
  (although the particular neighbour~$\variablesv_{\vertexu}$ will not
  contribute in the proof of the lower bound).
  If $\vertexv \neq \vertexu$,
  then
  setting all the
  variables~$x_{\vertexv, \vertexw} \in \variablesv_{\vertexv}$ to
  true and all the
  variables~$x_{\vertexw, \vertexv} \in \variablesv_{\vertexv}$ to
  false 
  (\ie making~$\vertexv$ into the minimal element of the set)
  satisfies $\formf_{\vertexu}$ as well as 
  all the affected axioms in~$\fixedconstraints$. 
  If~$\vertexv = \vertexu$, we can use a complementary assignment to
  the one above (\ie making~$\vertexv = \vertexu$ into the 
  maximal element of the set) to \radverb satisfy
  $\formf_{\vertexu}$. 
  Observe that this also shows that
  all $\variablesv_{\vertexv} \in \vpartition$ 
  \rsatrespect~$\fixedconstraints$ as required by
  \refdef{def:f-v-graph}.

  By the analysis above,
  it holds that
  the boundary $\boundary{\vertexsetv'}$ of some 
  vertex set $\vertexsetv'$ in $\graphg$ yields
  the \rboundary
  $
  \Sboundary{\Union_{\vertexu \in \vertexsetv'} \formf_\vertexu}
  \supseteq
  \setdescr{\variablesv_\vertexv}{\vertexv \in \boundary{\vertexsetv'}}
  $
  in~$\fvstructurenot$.
  Thus, 
  the expansion parameters for $\fvstructurenot$ are the same 
  as those for $\graphg$
  and we can conclude that $\fvstructurenot$ is 
  an~\rboundaryexpandershort{\expsize}{\expfactor}{0}{\fixedconstraints}.

  Finally, we note that while
  $\vpartition$ is \emph{not} a partition of the variables of~$\gop$,
  the \voverlapterm is only
  $\voverlapnot{\vpartition} = 2$ since  every variable
  $x_{\vertexu, \vertexv}$ occurs in exactly two 
  sets $\variablesv_{\vertexu}$ and $\variablesv_{\vertexv}$
  in~$\vpartition$. Hence, by
  \refcor{cor:FVStructureTheoremWithoutSlack} the degree of refuting
  $\gop$ is greater than $\expfactor \expsize / 4$.
\end{proof}
}

\ifthenelse{\boolean{conferenceversion}}{}{  With the previous theorem in hand, we can prove 
    (a version of)
  the main result in~\cite{GL10Optimality},
    namely that
  there exists a family of $5$-CNF formulas 
          witnessing that the lower bound on size in terms of degree in
  \refth{th:IPSPCTheorem} is essentially optimal. 
  That is, there are formulas over $N$~variables that can be refuted in
  polynomial calculus (in fact, in resolution) in size polynomial in~$N$
  but require degree~$\Bigomega{\sqrt{N}}$.
  This follows by plugging expanders with suitable parameters into
  \refth{th:Galesi-Lauria-theorem}. 
  By standard calculations (see, for example,~\cite{HLW06ExpanderGraphs})
  one can show that there exist constants $\expszfct, \expfactor > 0$
  such that randomly sampled graphs on $n$~vertices with degree at
  most~$5$ are $(\expszfct n, \expfactor)$-boundary expanders
  in the sense of \refdef{def:NonBipartiteBoundaryExpander} with high
  probability. By \refth{th:Galesi-Lauria-theorem}, graph ordering
  principle formulas on such graphs yield $5$-CNF formulas over
  $\Bigtheta{n^2}$~variables that require degree~$\bigomega{n}$. Since
  these formulas have polynomial calculus refutations in
  size~$\Bigoh{n^3}$ (just mimicking the resolution refutations
  constructed in~\cite{Stalmarck96ShortResolutionProofs}), this shows
  that the bound in \refth{th:IPSPCTheorem} is essentially tight.
  The difference between this bound and~\cite{GL10Optimality} is that
  since a weaker form of expansion is
  required
  in~\cite{GL10Optimality} it
  is possible to use $3$\nobreakdash-regular graphs, yielding families
  of \mbox{$3$-CNF} formulas.
      }

Let us now turn our attention back to bipartite graphs and consider
different flavours of pigeonhole principle formulas. We will 
focus on
formulas 
over
bounded-degree bipartite graphs,
where we will convert standard bipartite boundary expansion as
in~\refdef{def:bipartite-boundary-expander} into \rsatrespectful
boundary expansion as in \refdef{def:F-V-boundary-expansion}.
For a bipartite graph
$\graphg = (\vertexsetu \disunion \vertexsetv,\edgesete)$ 
the axioms appearing in the different versions of the graph pigeonhole
principle formulas are as follows:
\begin{subequations}
  \begin{align}
    \label{eq:PHPPigeonAxioms}
    & \Lor_{\vertexv \in \neighbor{\vertexu}} \varx_{\vertexu, \vertexv}
    & & 
    \vertexu \in \vertexsetu
    & & 
    \text{(pigeon axioms)} 
    \\
    \label{eq:PHPHoleAxioms}
    & 
    \olnot{\varx}_{\vertexu, \vertexv} \lor
    \olnot{\varx}_{\vertexu', \vertexv} 
    & &
    \vertexv \in \vertexsetv, \,
    \vertexu, \vertexu' \in \neighbor{\vertexv}, \,
    \vertexu \neq \vertexu', 
    & & 
    \text{(hole axioms)} 
    \\
    \label{eq:PHPFunctionality}
    & 
    \olnot{\varx}_{\vertexu, \vertexv} \lor 
    \olnot{\varx}_{\vertexu, \vertexv'} 
    & &  
    \vertexu \in \vertexsetu, \,
    \vertexv, \vertexv' \in \neighbor{\vertexu}, \, 
    \vertexv \neq \vertexv'
    & & 
    \text{(functionality axioms)} 
    \\
    \label{eq:PHPOntoAxioms}
    & \Lor_{\vertexu \in \neighbor{\vertexv}} \varx_{\vertexu, \vertexv}
    & & 
    \vertexv \in \vertexsetv
    & & 
    \text{(onto axioms)} 
  \end{align}
\end{subequations}
The ``plain vanilla''
\introduceterm{graph pigeonhole principle formula}~$\GraphPHPnot$
is the CNF formula over variables 
$\setdescr{\varx_{\vertexu, \vertexv}}{(\vertexu, \vertexv) \in \edgesete}$
consisting of clauses~\refeq{eq:PHPPigeonAxioms}
and~\refeq{eq:PHPHoleAxioms}; 
the \introduceterm{graph functional pigeonhole principle formula}
$\GraphFPHPnot$
contains the clauses of $\GraphPHPnot$ and in addition
clauses~\refeq{eq:PHPFunctionality};
the \introduceterm{graph onto pigeonhole principle formula}
$\GraphOntoPHPnot$  contains $\GraphPHPnot$ 
plus clauses~\refeq{eq:PHPOntoAxioms};
and the \introduceterm{graph onto functional pigeonhole principle formula}
$\GraphOntoFPHPnot$
consists of all  the
clauses~\refeq{eq:PHPPigeonAxioms}\nobreakdash--\refeq{eq:PHPOntoAxioms}. 

We obtain the standard versions of the PHP formulas by considering
graph formulas as above over the complete bipartite graph~$K_{n + 1, n}$.
In the opposite direction, for any bipartite graph~$\graphg$ 
with \mbox{$n+1$~vertices} on the left and
$n$~vertices on the right we can hit any version of the pigeonhole
principle formula over~$K_{n + 1, n}$ 
with the  restriction~$\rstd_\graphg$ setting
$\varx_{\vertexu,\vertexv}$ to false for all
$(\vertexu,\vertexv) \notin \edgesete(\graphg)$
to recover the corresponding graph pigeonhole principle formula
over~$\graphg$. 
When doing so, we will use the observation from 
\refsec{sec:preliminaries}
that restricting a formula can only decrease the size and degree
required to refute it.

As mentioned in \refsec{sec:intro}, it
was established already in \cite{AR03LowerBounds} that 
good bipartite boundary expanders~$\graphg$ yield
formulas $\graphphpnot$ that require large
polynomial calculus degree
to refute. We can reprove this result in our language---and, in fact, 
observe that the lower bound in~\cite{AR03LowerBounds} works also for
the onto version~$\GraphOntoPHPnot$---by
constructing an appropriate \fvstructure. In addition, we can
generalize the result in~\cite{AR03LowerBounds}  slightly by allowing some 
additive slack $\expslack > 0$ in the 
expansion in \refth{th:FVStructureTheorem}.
This works as long as we have the guarantee that  
no too small subformulas are unsatisfiable.

\begin{theorem}
  \label{th:mn14-php-theorem}
  Suppose that
  $\graphg = (\vertexsetu \disunion \vertexsetv, \edgesete)$ 
  is a bipartite graph with
  $\setsize{\vertexsetu} = \paramn$ and $\setsize{\vertexsetv} = \paramn - 1$
  and that
  $\expfactor > 0$ is a constant such that
  \begin{itemize}

  \item for every set $\vertexsetu' \subseteq \vertexsetu$ of size
    $\setsize{\vertexsetu'} \leq \expsize$ there is a matching of
    $\vertexsetu'$ into $\vertexsetv$, and

  \item for every set $\vertexsetu' \subseteq \vertexsetu$ of size
    $\setsize{\vertexsetu'} \leq \expsize$ it holds that
    $\setsize{\boundary{\vertexsetu'}} \geq \expfactor
    \setsize{\vertexsetu'} - \expslack$.
  \end{itemize}
  Then
  $\mdegreeref{\GraphOntoPHPnot} > 
  \expfactor \expsize / 2 - \expslack$.
\end{theorem}

\begin{proof}[Proof sketch]
  The \fvstructure 
  for $\graphphpnot$ 
  is formed by taking $\fpartition$ to be 
  the
  set of
  pigeon axioms~\eqref{eq:PHPPigeonAxioms}, 
  $\fixedconstraints$~to consist of
  the hole axioms~\eqref{eq:PHPHoleAxioms}
  and onto axioms~\eqref{eq:PHPOntoAxioms}, and
  $\vpartition$~to be the collection of variable sets
  $\variablesv_{\vertexv} = \setdescr{\varx_{\vertexu,
      \vertexv}}{\vertexu \in \neighbor{\vertexv}}$
  partitioned with respect to the holes~$\vertexv \in \vertexsetv$. It
  is straightforward to 
    check
  that this \fvstructure is isomorphic to the graph $\graphg$
  and that all neighbours in
  $\fvstructurenot$ are \radjective
  (for
  $\Lor_{\vertexv \in \neighbor{\vertexu}} \varx_{\vertexu, \vertexv}
  \in  \fpartition$
  and
  $\variablesv_\vertexv$ for some
  $\vertexv \in \neighbor{\vertexu}$,
  apply the partial assignment sending
  pigeon~$\vertexu$ to hole~$\vertexv$
  and ruling out all other pigeons in~$\neighbor{\vertexv} \setminus
  \set{\vertexu}$ for~$\vertexv$).
  Moreover, using the existence of
  matchings for all sets of pigeons~$\vertexsetu'$ of size
  $\setsize{\vertexsetu'} \leq \expsize$ we can prove 
  that every subformula $\fpartition' \land \fixedconstraints$ is
  satisfiable as long as~$\setsize{\fpartition'} \leq
  \expsize$. Hence, we can apply 
  \refth{th:FVStructureTheorem} 
  to derive the claimed bound.
  We refer to the upcoming full-length version of~\cite{MN14LongProofs}
  for the omitted details.
\end{proof}

\Refth{th:mn14-php-theorem} is the only place in this paper where we
use non-zero slack for the expansion. 
The reason that we need slack is so that we can
establish lower bounds for another type of formulas, namely the subset
cardinality formulas studied in
\ifthenelse{\boolean{conferenceversion}}
{\cite{MN14LongProofs,Spence10Sgen1,VS10ZeroOneDesigns}.} 
{\cite{Spence10Sgen1,VS10ZeroOneDesigns,MN14LongProofs}.} 
A brief (and somewhat informal) description of these formulas is as
follows. We start with a $4$-regular
bipartite graph to which we add an extra edge between two
non-connected vertices.
We then write down clauses stating that each \mbox{degree-$4$}
vertex on the left has at least $2$~of its edges set to true, while
the single degree-$5$ vertex has a strict majority of $3$~incident
edges set to true. On the right-hand side of the graph we 
encode the opposite, namely that all vertices with degree~$4$ have at
least $2$~of its edges set to false, while the vertex with degree~$5$
has at least $3$~edges set to false. \mbox{A simple} counting argument 
yields that the CNF formula consisting of these clauses must be
unsatisfiable. Formally, we have the following definition (which
strictly speaking is a slightly specialized case of the general
construction, but again we refer to~\cite{MN14LongProofs} for the
details).

\begin{definition}[Subset cardinality formulas 
    \ifthenelse{\boolean{conferenceversion}}    {\cite{MN14LongProofs, VS10ZeroOneDesigns}}    {\cite{VS10ZeroOneDesigns, MN14LongProofs}}  ]
  \label{def:sv-formulas}
  Suppose that
  $\graphg = (\vertexsetu \disunion \vertexsetv, \edgesete)$ is a
  bipartite graph that is $4$\nobreakdash-regular except that one
  extra edge has been added between two unconnected vertices on the
  left and right.
  Then the \introduceterm{subset cardinality formula}
  $\subcardnot{\graphg}$ over $\graphg$ has variables
  $\varx_{\edgee}, \edgee \in \edgesete$, and clauses:
  \begin{itemize}
  \item $\varx_{\edgee_1} \lor \varx_{\edgee_2} \lor \varx_{\edgee_3}$
    for every triple $\edgee_1, \edgee_2, \edgee_3$ of edges incident
    to any $\vertexu \in \vertexsetu$,
  \item $\olnot{\varx}_{\edgee_1} \lor \olnot{\varx}_{\edgee_2} \lor
    \olnot{\varx}_{\edgee_3}$ for every triple $\edgee_1, \edgee_2,
    \edgee_3$ of edges incident to any $\vertexv \in \vertexsetv$.
  \end{itemize}
\end{definition}

To prove lower bounds on refutation degree for these formulas
we use the standard notion of vertex
expansion on bipartite graphs, where all neighbours on the left are
counted and not just unique neighbours as in
\refdef{def:bipartite-boundary-expander}.

\begin{definition}[Bipartite expander]
  A bipartite graph
  $\graphg = (\vertexsetu \disunion \vertexsetv, \edgesete)$ is a
  \introduceterm{bipartite $(\expsize, \expfactor)$-expander} if for
  each vertex set
  $\vertexsetu' \subseteq \vertexsetu, \setsize{\vertexsetu'} \leq
  \expsize$,
  it holds that
  $\setsize{\neighbor{\vertexsetu'}} \geq \expfactor
  \setsize{\vertexsetu'}$.
\end{definition}

The existence of such expanders with appropriate parameters can again
be established by straightforward calculations (as in, for
instance,~\cite{HLW06ExpanderGraphs}). 

\begin{theorem}[\cite{MN14LongProofs}]
  \label{th:mn-subset-cardinality-lb}
  Suppose that
  $\graphg = (\vertexsetu \disunion \vertexsetv, \edgesete)$ is a
  $4$-regular bipartite
  $\bigl(\expszfct \paramn, \frac{5}{2} + \expfactor\bigr)$-expander
  for \mbox{$\setsize{\vertexsetu} = \setsize{\vertexsetv} = \paramn$}
  and some constants $\expszfct, \expfactor > 0$, and let $\graphg'$ be
  obtained from~$\graphg$ by adding an arbitrary edge 
  between two unconnected vertices in~$\vertexsetu$ and~$\vertexsetv$. 
  Then refuting
  the formula
  $\subcardnot{\graphg'}$ requires degree
  $\mdegreeref{\subcardnot{\graphg'}} =\bigomega{n}$, 
  and hence size
  $\sizeref[\pcrnot]{\subcardnot{\graphg'}} =
  \exp\bigl(\bigomega{n}\bigr)$. 
\end{theorem}

\begin{proof}[Proof sketch]
  The proof is by reducing to graph PHP
  formulas and applying \refth{th:mn14-php-theorem}
  (which of course also holds with onto axioms removed). 
  We fix some complete matching in~$\graphg$, 
  which is guaranteed to exist in
  regular bipartite graphs, and then set all edges in the matching
  as well as the extra added edge to true. Now the \mbox{degree-$5$} vertex
  $\vertexv^*$ on 
  the right has only $3$~neighbours and 
  the constraint for~$\vertexv^*$ requires all of these edges to be
  set to false. Hence, we set these edges to false as well 
  which makes
  $\vertexv^*$ and its clauses vanish from the formula. The
  restriction leaves us with $n$~vertices on the left which require
  that at least $1$~of the remaining $3$~edges incident to them is
  true, while the \mbox{$n-1$ vertices} on the right require that at
  most $1$~out of  their incident edges is true. 
  That is, we have restricted our subset cardinality formula to obtain
  a graph PHP formula.

  As the original graph is a
  $(\expszfct \paramn, \frac{5}{2} + \expfactor)$-expander, a simple
  calculation can convince us that the new graph is a boundary
  expander where each set of vertices $\vertexsetu'$ on the left with
  size $\setsize{\vertexsetu'} \leq \expszfct \paramn$ has 
  boundary expansion
  $\setsize{\boundary{\vertexsetu'}} \geq 2 \expfactor
  \setsize{\vertexsetu'} - 1$.
  Note the additive slack of~$1$ compared to the usual expansion
  condition, which 
  is caused by the
  removal of the \mbox{degree-$5$} vertex
  $\vertexv^*$ from the right. 
  Now we can appeal to \refth{th:mn14-php-theorem} 
  (and \refth{th:IPSPCTheorem})
  to obtain the lower bounds claimed in the theorem.
\end{proof}

Let us conclude this section by presenting our new lower bounds for
the functional pigeonhole principle formulas.
As a first attempt, we could try to reason as in the proof of
\refth{th:mn14-php-theorem} 
 (but adding the
axioms~\refeq{eq:PHPFunctionality} and removing
axioms~\eqref{eq:PHPOntoAxioms}).
The naive idea would be to modify our \fvstructure slightly by
substituting the functionality axioms for the onto axioms
in~$\fixedconstraints$ while keeping 
$\fpartition$ and~$\vpartition$ the same.
This does not work, however---although the sets
$\variablesv_{\vertexv} \in \vpartition$ are \radjective, the only
assignment that \rsatrespects $\fixedconstraints$ is the one
that sets all variables $\varx_{\vertexu, \vertexv} \in
\variablesv_{\vertexv}$ to false.
Thus, it is not possible to satisfy any of the pigeon axioms, meaning
that there are no \radjective neighbours in $\fvstructurenot$.
In order to obtain a useful \fvstructure, we instead need to redefine
$\vpartition$ by enlarging the variable sets $\variablesv_{\vertexv}$,
using the fact that $\vpartition$ is not required to be a partition.
Doing so in the appropriate way yields the following theorem.

\begin{theorem}
  \label{th:FPHPDegreeLowerBound}
  Suppose that
  $\graphg = (\vertexsetu \disunion \vertexsetv, \edgesete)$
  is a bipartite $(\expsize, \expfactor)$-boundary expander
  with left degree bounded by~$\graphdeg$. Then it holds that refuting
  $\GraphFPHPnot[\graphg]$
  in polynomial calculus requires degree strictly greater than
  $\expfactor \expsize / (2 \graphdeg)$.
  It follows that 
  if $\graphg$ is a bipartite
  $(\expszfct \paramn, \expfactor)$-boundary expander with
  constant left degree \mbox{and $\expszfct, \expfactor > 0$}, then any
  polynomial calculus (PC or PCR)
  refutation of $\GraphFPHPnot[\graphg]$ requires size
  $\exp(\bigomega{\paramn})$.
\end{theorem}

\begin{proof}
  We construct \afvstructure from $\GraphFPHPnot[\graphg]$ as follows.
  We let 
  the set of clauses
  $\fixedconstraints$ consist of all
  hole axioms~\refeq{eq:PHPHoleAxioms} and functionality
  axioms~\refeq{eq:PHPFunctionality}.
  We define the family $\fpartition$
  to consist of the pigeon axioms~\refeq{eq:PHPPigeonAxioms}
  interpreted as singleton CNF formulas.
  For the variables we let
  $
  \vpartition = 
  \setdescr{\variablesv_{\vertexv}}{\vertexv \in \vertexsetv}
  $,
  where 
  for every hole
  $\vertexv \in \vertexsetv$
  the set~$\variablesv_{\vertexv}$ is defined by
  \begin{equation}
    \label{eq:fphp-variable-sets}
    \variablesv_{\vertexv} = 
    \Setdescr{\varx_{\vertexu', \vertexv'}} 
             {\text{$\vertexu' \in \neighbor{\vertexv}$ and
                 $\vertexv' \in \neighbor{\vertexu'}$}}
             \eqperiod
  \end{equation}
  That is, to build   $\variablesv_{\vertexv}$ 
  we start with the hole~$\vertexv$ on the right, 
  consider all pigeons~$\vertexu'$ on the left that can go into this
  hole,
  and finally include in
  $\variablesv_{\vertexv}$ 
  for all such~$\vertexu'$
  the variables
  $\varx_{\vertexu', \vertexv'}$
  for all holes~$\vertexv'$
  incident to~$\vertexu'$.
  We want to show that $\fvstructurenot$ as defined above satisfies the
  conditions in \refcor{cor:FVStructureTheoremWithoutSlack}.

  Note first that every variable set~$\variablesv_{\vertexv}$ 
  respects the  
  clause set~$\fixedconstraints$ since
  setting all variables in~$\variablesv_{\vertexv}$ to false satisfies
  all clauses in~$\fixedconstraints$ mentioning variables
  in~$\variablesv_{\vertexv}$. 
  It is easy to see from~\refeq{eq:fphp-variable-sets}
  that when a hole $\vertexv$ is a neighbour of a pigeon~$\vertexu$,
  the variable set~$\variablesv_{\vertexv}$ is also a neighbour 
  in
  the \fvstructure of the
  corresponding pigeon axiom 
  $
  \formf_{\vertexu} = 
  \Lor_{\vertexv \in \neighbor{\vertexu}} \varx_{\vertexu, \vertexv}$.
  These are the only neighbours of the
  pigeon axiom~$\formf_{\vertexu}$, as each $\variablesv_{\vertexv}$
  contains only variables mentioning pigeons in the neighbourhood
  of~$\vertexv$. In other words, $\graphg$ and~$\fvstructurenot$ share
  the same neighbourhood structure.   

  Moreover,  we claim that 
  every neighbour $\variablesv_{\vertexv}$ of
  $\formf_{\vertexu}$ is an \rneighbour. To see this, consider the
  assignment~$\assignmenta_{\vertexu, \vertexv}$ that sets
  $\varx_{\vertexu, \vertexv}$ to true and the remaining variables in
  $\variablesv_{\vertexv}$ to false. Clearly, $\formf_{\vertexu}$ is
  satisfied by~$\assignmenta_{\vertexu, \vertexv}$.
  All axioms in~$\fixedconstraints$ \emph{not} containing
  $\varx_{\vertexu, \vertexv}$ are either 
  satisfied by
  $\assignmenta_{\vertexu, \vertexv}$
  or left untouched, since   $\assignmenta_{\vertexu, \vertexv}$ assigns
  all other variables in its domain to false. 
  Any hole axiom
  $\olnot{\varx}_{\vertexu, \vertexv} \lor
  \olnot{\varx}_{\vertexu', \vertexv}$
  in~$\fixedconstraints$
  that \emph{does} contain~$\varx_{\vertexu, \vertexv}$ is satisfied
  by~$\assignmenta_{\vertexu, \vertexv}$ since 
  ${\varx}_{\vertexu', \vertexv} \in \variablesv_\vertexv$
  for
  $\vertexu' \in \neighbor{\vertexv}$ by~\refeq{eq:fphp-variable-sets}
  and
  this variable is set to false   by~$\assignmenta_{\vertexu, \vertexv}$.
  In the same way, any functionality axiom
  $
  \olnot{\varx}_{\vertexu, \vertexv} \lor 
  \olnot{\varx}_{\vertexu, \vertexv'} 
  $
  containing
  $\varx_{\vertexu, \vertexv}$ is satisfied since 
  the variable
  ${\varx}_{\vertexu, \vertexv'}$ is in~$\variablesv_\vertexv$
  by~\refeq{eq:fphp-variable-sets}
  and is hence assigned to false.
  Thus, the assignment~$\assignmenta_{\vertexu, \vertexv}$
  \rsatisfies~$\formf_{\vertexu}$, 
  and so 
  $\formf_{\vertexu}$ and~$\variablesv_{\vertexv}$ are
  \rneighbours as claimed.

  Since our constructed \fvstructure is isomorphic to the
  original graph~$\graphg$ and all neighbour relations
  are \rsatrespectful, 
  the expansion parameters of~$\graphg$ trivially carry over to
  \rsatrespectful expansion in~$\fvstructurenot$. This is just
  another way of saying that $\fvstructurenot$ is an
  \rboundaryexpandershort{\expsize}{\expfactor}{0}{\fixedconstraints}.

  To finish the proof, note that
  the \voverlapterm of $\vpartition$ is
  at most $\graphdeg$. This is so since a variable~$\varx_{\vertexu, \vertexv}$
  appears in a set~$\variablesv_{\vertexv'}$ only when
  $\vertexv' \in \neighbor{\vertexu}$. Hence, for all
  variables~$\varx_{\vertexu, \vertexv}$ it holds that they appear in
  at most
  \mbox{$\setsize{\neighbor{\vertexu}} \leq \graphdeg$ sets}
  in~$\vpartition$.
  Now the conclusion that
  any polynomial calculus refutation of
  $\GraphFPHPnot[\graphg]$ requires degree greater than
  $\expfactor \expsize / (2 \graphdeg)$
  can be read off from
  \refcor{cor:FVStructureTheoremWithoutSlack}.
  In addition, the exponential lower bound on the size of 
  a refutation of $\graphfphpnot$ when $\graphg$ is a 
  $(\expszfct \paramn, \expfactor)$-boundary expander $\graphg$ with
  constant left degree follows by 
  plugging the degree lower bound into~\refth{th:IPSPCTheorem}.
\end{proof}

It is not hard to show (again we refer to~\cite{HLW06ExpanderGraphs}
for the details) that there exist bipartite graphs with
left degree~$3$ 
which are $(\expszfct \paramn, \expfactor)$-boundary
expanders for $\expszfct, \expfactor > 0$ and hence our size lower bound 
for polynomial calculus refutations
of~$\GraphFPHPnot$ can be applied to them.
Moreover, if 
$\setsize{\vertexsetu} = n+1$
and~$\setsize{\vertexsetv} = n$,
then we can identify some bipartite graph $\graphg$
that is a good expander and
hit
$\fphpnot{n + 1}{n}
=
\GraphFPHPnot[K_{n + 1, n}]$ 
with a restriction~$\rstd_\graphg$ setting
$\varx_{\vertexu,\vertexv}$ to false for all
$(\vertexu,\vertexv) \notin \edgesete$
to obtain
$
\restrict{\fphpnot{n + 1}{n}}{\rstd_\graphg}
=
\GraphFPHPnot$.
Since restrictions can only decrease refutation size, it follows that
size lower bounds for~$\GraphFPHPnot$ apply also to
$\fphpnot{n + 1}{n}$,
yielding the second lower bound claimed in \refsec{sec:our-results}.

\begin{theorem}
  \label{th:main-result}
  Any polynomial calculus or polynomial calculus resolution
  refutation of 
  (the standard CNF encoding of) the functional pigeonhole principle
  $\fphpnot{n + 1}{n}$ requires size~$\exp(\bigomega{n})$.
\end{theorem}

\makeatletter{}\section{Concluding Remarks}
\label{sec:conclusion}

In this work, we extend the techniques 
developed by Alekhnovich and Razborov~\cite{AR03LowerBounds} for
proving degree lower bounds on refutations of CNF formulas
in polynomial calculus. Instead of looking
at the clause-variable incidence graph~$G(F)$ of the formula~$F$ as in
\cite{AR03LowerBounds},
we allow clustering of clauses and variables and reason in terms of
the incidence graph $G'$ defined on these clusters.  
We show that the CNF formula~$F$ requires high degree 
to be refuted
in polynomial calculus
whenever this clustering can be done in a way that ``respects the
structure'' of the formula and 
so that
the resulting graph~$G'$ has certain
expansion properties.

This 
provides us with a unified framework within which we can
reprove previously established  degree lower bounds
in~\cite{AR03LowerBounds,GL10Optimality,MN14LongProofs}. 
More importantly, this also allows us to 
obtain
a degree lower
bound on the functional pigeonhole principle defined on expander
graphs, solving an open problem from~\cite{Razborov02ProofComplexityPHP}.
It immediately follows from this that the (standard CNF encodings of)
the usual functional pigeonhole principle formulas require
exponential proof size in polynomial calculus resolution,
resolving a question on Razborov's problems
list~\cite{Razborov14webpage}
which had (quite annoyingly) remained open.
This means that we now have an essentially complete understanding of
how the different variants of pigeonhole principle formulas behave
with respect to polynomial calculus in the standard setting with
\mbox{$n+1$ pigeons} and \mbox{$n$ holes}. 
Namely, while Onto-FPHP formulas are easy, 
both FPHP formulas and Onto-PHP formulas are exponentially hard in~$n$
even when restricted to bounded-degree expanders.

A natural next step would be to see if this generalized framework can
also be used to attack other interesting formula families 
which are known to be hard for resolution but
for which there are
currently no lower bounds in polynomial calculus.
In
particular, can our framework or some modification of it prove a lower
bound for refuting the formulas encoding that a graph does not contain
an independent set of size~$k$, which were proven hard for resolution
in~\cite{BIS07IndependentSets}? 
Or
what about the formulas stating that a
graph is $k$-colorable, for which resolution lower bounds were
established in~\cite{BCCM05RandomGraph}?

Returning to the pigeonhole principle, we now understand 
how different
encodings behave in polynomial calculus when we have $n + 1$ pigeons
and $n$ holes. But what happens when we increase the number
of pigeons? For instance, do the  formulas become easier if we have $n^2$
pigeons and $n$ holes? (This is the point where lower bound techniques
based on degree break down.) What about arbitrary many pigeons? In
resolution these questions are fairly well understood, as witnessed
by the works of
Raz~\cite{Raz04Resolution} and
Razborov~\cite{Razborov01ImprovedResolutionLowerBoundsWPHP,Razborov03ResolutionLowerBoundsWFPHP,Razborov04ResolutionLowerBoundsPM}, 
but as far as we are aware they remain wide open for polynomial calculus.

Finally, we want to point out an intriguing contrast between our work
and that of Alekhnovich and Razborov.  As discussed in the
introduction, the main technical result in~\cite{AR03LowerBounds}
is that when the incidence graph of a set of polynomial equations is
expanding and the polynomials are immune, \ie have no low-degree
consequences, then refuting this set of equations is hard \wrt
polynomial calculus degree. Since clauses of width~$w$ have maximal
immunity~$w$, it follows that for a CNF formula~$\formf$ expansion of
the clause-variable incidence graph $\graphg(\formf)$ is enough to
imply hardness. A natural way of interpreting our work would be to say
that we simply extend this result to a slightly more general
constraint-variable incidence graph. On closer inspection, however,
this analogy seems to be misleading, and since we were quite surprised
by this ourselves we want to elaborate briefly on this.

For the functional pigeonhole principle, the pigeon and functional
axioms for a pigeon~$\vertexu$ taken together imply the polynomial
equation
$\sum_{\vertexv \in \neighbor{\vertexu}} \varx_{\vertexu,\vertexv} = 1$
(summing over all holes $\vertexv \in \neighbor{\vertexu}$ to which the
pigeon~$\vertexu$ can fly).  Since this is a \mbox{degree-$1$}
consequence, it shows that the pigeonhole axioms in FPHP formulas have
\emph{lowest possible immunity} modulo the set~$\fixedconstraints$ consisting
of hole and functionality axioms. Nevertheless, our lower bound proof
still works, and only needs expansion of the constraint-variable graph
although the immunity of the constraints is non-existent.

On the other hand, the constraint-variable incidence graph of a random
set of parity constraints is expanding asymptotically almost surely,
and since over fields of characteristic distinct from~$2$ parity
constraints have high immunity (see,
\eg,~\cite{Green00ComplexNumberFourier}), the techniques
in~\cite{AR03LowerBounds} can be used to prove strong degree lower
bounds in such a setting. However, it seems that our framework of
\rsatrespectful boundary expansion is inherently unable to establish
this result. The problem is that (as discussed in the footnote after
\refdef{def:F-V-boundary-expansion}) it is not possible to group
variables together in such a way as to ensure \rsatrespectful
neighbourhood relations. At a high level, it seems that the main
ingredient needed for our technique to work is that
clauses/polynomials and variables can be grouped together in such a
way that the effects of assignments to a group of variables can always
be contained in a small neighbourhood of clauses/polynomials, which
the assignments (mostly) satisify, and do not propagate beyond this
neighbourhood. Functional pigeonhole principle formulas over bounded-degree
graphs have this property, since assigning a pigeon~$\vertexu$ to a
hole~$\vertexv$ only affects the neighbouring holes of~$\vertexu$ and
the neighbouring pigeons of~$\vertexv$, respectively. There is no such
way to contain the effects locally when one starts satisfying
individual equations in an expanding set of parity constraints,
however, regardless of the characteristic of the underlying field.

In view of this, it seems that our techniques and those
of~\cite{AR03LowerBounds} are closer to being orthogonal rather than
parallel.  It would be desirable to gain a deeper understanding of
what is going on here. In particular, in comparison
to~\cite{AR03LowerBounds}, which gives clear, explicit criteria for
hardness (is the graph expanding? are the polynomials immune?), our
work is less explicit in that it says that hardness is implied by the
existence of a ``clustered clause-variable incidence graph'' with the
right properties, but gives no guidance as to if and how such a graph
might be built.  It would be very interesting to find more general
criteria of hardness that could capture both our approach and that
of~\cite{AR03LowerBounds}, and ideally provide a unified view of these
lower bound techniques.

\makeatletter{}
\ifthenelse{\boolean{conferenceversion}}
{\subparagraph*{Acknowledgements}}
{\section*{Acknowledgements}}

We are grateful to Ilario~Bonacina, Yuval~Filmus, Nicola~Galesi,
Massimo~Lauria, Alexander~Razborov, and Marc~Vinyals for numerous
discussions on proof complexity in general and polynomial calculus
degree lower bounds in particular. We want to give a special thanks to
Massimo~Lauria for several insightful comments on an earlier version
of this work, which allowed us to simplify the construction (and improve
the parameters in the results) considerably, and to Alexander~Razborov
for valuable remarks on a preliminary version of this manuscript that,
in particular, helped to shed light on the similarities with and
differences from the techniques in~\cite{AR03LowerBounds}. Finally, we
are thankful for the feedback provided by the anonymous
\mbox{\emph{CCC~'15}} referees and participants of the Dagstuhl
workshop~15171  \emph{Theory and Practice of SAT Solving}
in April 2015.

The authors were funded by the
European Research Council under the European Union's Seventh Framework
Programme \mbox{(FP7/2007--2013) /} ERC grant agreement no.~279611.
\TheauthorJN 
was also supported by
Swedish Research Council grants 
\mbox{621-2010-4797}
and
\mbox{621-2012-5645}.

\bibliography{PCdegreeFPHP_SingleFile}

\begin{thebibliography}{BCMM05}

\bibitem[ABRW02]{ABRW02SpaceComplexity}
Michael Alekhnovich, Eli {Ben-Sasson}, Alexander~A. Razborov, and Avi
  Wigderson.
\newblock Space complexity in propositional calculus.
\newblock {\em SIAM Journal on Computing}, 31(4):1184\nobreakdash--1211, 2002.
\newblock Preliminary version appeared in \emph{STOC~'00}.

\bibitem[AR03]{AR03LowerBounds}
Michael Alekhnovich and Alexander~A. Razborov.
\newblock Lower bounds for polynomial calculus: {N}on-binomial case.
\newblock {\em Proceedings of the Steklov Institute of Mathematics},
  242:18\nobreakdash--35, 2003.
\newblock Available at
  \url{http://people.cs.uchicago.edu/~razborov/files/misha.pdf}. Preliminary
  version appeared in \emph{FOCS~'01}.

\bibitem[BCMM05]{BCCM05RandomGraph}
Paul Beame, Joseph~C. Culberson, David~G. Mitchell, and Cristopher Moore.
\newblock The resolution complexity of random graph $k$-colorability.
\newblock {\em Discrete Applied Mathematics},
  153(1\nobreakdash-3):25\nobreakdash--47, December 2005.

\bibitem[BGIP01]{BGIP01LinearGaps}
Samuel~R. Buss, Dima Grigoriev, Russell Impagliazzo, and Toniann Pitassi.
\newblock Linear gaps between degrees for the polynomial calculus modulo
  distinct primes.
\newblock {\em Journal of Computer and System Sciences},
  62(2):267\nobreakdash--289, March 2001.
\newblock Preliminary version appeared in \emph{CCC~'99}.

\bibitem[BIS07]{BIS07IndependentSets}
Paul Beame, Russell Impagliazzo, and Ashish Sabharwal.
\newblock The resolution complexity of independent sets and vertex covers in
  random graphs.
\newblock {\em Computational Complexity}, 16(3):245\nobreakdash--297, October
  2007.

\bibitem[Bla37]{Blake37Thesis}
Archie Blake.
\newblock {\em Canonical Expressions in {B}oolean Algebra}.
\newblock PhD thesis, University of Chicago, 1937.

\bibitem[BW01]{BW01ShortProofs}
Eli {Ben-Sasson} and Avi Wigderson.
\newblock Short proofs are narrow---resolution made simple.
\newblock {\em Journal of the ACM}, 48(2):149\nobreakdash--169, March 2001.
\newblock Preliminary version appeared in \emph{STOC~'99}.

\bibitem[CEI96]{CEI96Groebner}
Matthew Clegg, Jeffery Edmonds, and Russell Impagliazzo.
\newblock Using the {Groebner} basis algorithm to find proofs of
  unsatisfiability.
\newblock In {\em Proceedings of the 28th Annual {ACM} Symposium on Theory of
  Computing ({STOC}~'96)}, pages 174\nobreakdash--183, May 1996.

\bibitem[CR79]{CR79Relative}
Stephen~A. Cook and Robert Reckhow.
\newblock The relative efficiency of propositional proof systems.
\newblock {\em Journal of Symbolic Logic}, 44(1):36\nobreakdash--50, March
  1979.

\bibitem[CS88]{CS88ManyHard}
Va{\v{s}}ek Chv{\'a}tal and Endre Szemer{\'e}di.
\newblock Many hard examples for resolution.
\newblock {\em Journal of the ACM}, 35(4):759\nobreakdash--768, October 1988.

\bibitem[Fil14]{Filmus14OnAlekhnovichRazborovDegreeLowerBound}
Yuval Filmus.
\newblock On the {A}lekhnovich--{R}azborov degree lower bound for the
  polynomial calculus.
\newblock Manuscript. Available at
  \url{http://www.cs.toronto.edu/~yuvalf/AlRa.pdf}, 2014.

\bibitem[GL10]{GL10Optimality}
Nicola Galesi and Massimo Lauria.
\newblock Optimality of size-degree trade-offs for polynomial calculus.
\newblock {\em ACM Transactions on Computational Logic},
  12:4:1\nobreakdash--4:22, November 2010.

\bibitem[Gre00]{Green00ComplexNumberFourier}
Frederic Green.
\newblock A complex-number {F}ourier technique for lower bounds on the
  mod\nobreakdash-$m$ degree.
\newblock {\em Computational Complexity}, 9(1):16\nobreakdash--38, January
  2000.

\bibitem[Gri98]{Grigoriev98Tseitin}
Dima Grigoriev.
\newblock Tseitin's tautologies and lower bounds for {Nullstellensatz} proofs.
\newblock In {\em Proceedings of the 39th Annual {IEEE} Symposium on
  Foundations of Computer Science ({FOCS}~'98)}, pages 648\nobreakdash--652,
  November 1998.

\bibitem[Hak85]{Haken85Intractability}
Armin Haken.
\newblock The intractability of resolution.
\newblock {\em Theoretical Computer Science},
  39(2\nobreakdash-3):297\nobreakdash--308, August 1985.

\bibitem[HLW06]{HLW06ExpanderGraphs}
Shlomo Hoory, Nathan Linial, and Avi Wigderson.
\newblock Expander graphs and their applications.
\newblock {\em Bulletin of the American Mathematical Society},
  43(4):439\nobreakdash--561, October 2006.

\bibitem[IPS99]{IPS99LowerBounds}
Russell Impagliazzo, Pavel Pudl{\'a}k, and Ji{\v{r}}\'i Sgall.
\newblock Lower bounds for the polynomial calculus and the {G}r{\"o}bner basis
  algorithm.
\newblock {\em Computational Complexity}, 8(2):127\nobreakdash--144, 1999.

\bibitem[MN14]{MN14LongProofs}
Mladen Mik\v{s}a and Jakob Nordström.
\newblock Long proofs of (seemingly) simple formulas.
\newblock In {\em Proceedings of the 17th International Conference on Theory
  and Applications of Satisfiability Testing ({SAT}~'14)}, volume 8561 of {\em
  Lecture Notes in Computer Science}, pages 121\nobreakdash--137. Springer,
  July 2014.

\bibitem[Nor13]{Nordstrom13SurveyLMCS}
Jakob Nordström.
\newblock Pebble games, proof complexity and time-space trade-offs.
\newblock {\em Logical Methods in Computer Science}, 9:15:1\nobreakdash--15:63,
  September 2013.

\bibitem[Raz98]{Razborov98LowerBound}
Alexander~A. Razborov.
\newblock Lower bounds for the polynomial calculus.
\newblock {\em Computational Complexity}, 7(4):291\nobreakdash--324, December
  1998.

\bibitem[Raz01]{Razborov01ImprovedResolutionLowerBoundsWPHP}
Alexander~A. Razborov.
\newblock Improved resolution lower bounds for the weak pigeonhole principle.
\newblock Technical Report TR01-055, Electronic Colloquium on Computational
  Complexity (ECCC), July 2001.

\bibitem[Raz02]{Razborov02ProofComplexityPHP}
Alexander~A. Razborov.
\newblock Proof complexity of pigeonhole principles.
\newblock In {\em 5th International Conference on Developments in Language
  Theory, ({DLT}~'01), Revised Papers}, volume 2295 of {\em Lecture Notes in
  Computer Science}, pages 100\nobreakdash--116. Springer, July 2002.

\bibitem[Raz03]{Razborov03ResolutionLowerBoundsWFPHP}
Alexander~A. Razborov.
\newblock Resolution lower bounds for the weak functional pigeonhole principle.
\newblock {\em Theoretical Computer Science}, 1(303):233\nobreakdash--243, June
  2003.

\bibitem[Raz04a]{Raz04Resolution}
Ran Raz.
\newblock Resolution lower bounds for the weak pigeonhole principle.
\newblock {\em Journal of the ACM}, 51(2):115\nobreakdash--138, March 2004.
\newblock Preliminary version appeared in \emph{STOC~'02}.

\bibitem[Raz04b]{Razborov04ResolutionLowerBoundsPM}
Alexander~A. Razborov.
\newblock Resolution lower bounds for perfect matching principles.
\newblock {\em Journal of Computer and System Sciences},
  69(1):3\nobreakdash--27, August 2004.
\newblock Preliminary version appeared in \emph{CCC~'02}.

\bibitem[Raz15]{Razborov14webpage}
Alexander Razborov.
\newblock Possible research directions.
\newblock List of open problems (in proof complexity and other areas) available
  at \url{http://people.cs.uchicago.edu/~razborov/teaching/}, 2015.

\bibitem[Rii93]{Riis93Thesis}
S\o{}ren Riis.
\newblock {\em Independence in Bounded Arithmetic}.
\newblock PhD thesis, University of Oxford, 1993.

\bibitem[Spe10]{Spence10Sgen1}
Ivor Spence.
\newblock sgen1: {A} generator of small but difficult satisfiability
  benchmarks.
\newblock {\em Journal of Experimental Algorithmics},
  15:1.2:1.1\nobreakdash--1.2:1.15, March 2010.

\bibitem[St{\aa}96]{Stalmarck96ShortResolutionProofs}
Gunnar St{\aa}lmarck.
\newblock Short resolution proofs for a sequence of tricky formulas.
\newblock {\em Acta Informatica}, 33(3):277\nobreakdash--280, May 1996.

\bibitem[Urq87]{Urquhart87HardExamples}
Alasdair Urquhart.
\newblock Hard examples for resolution.
\newblock {\em Journal of the ACM}, 34(1):209\nobreakdash--219, January 1987.

\bibitem[VS10]{VS10ZeroOneDesigns}
Allen {Van Gelder} and Ivor Spence.
\newblock Zero-one designs produce small hard {SAT} instances.
\newblock In {\em Proceedings of the 13th International Conference on Theory
  and Applications of Satisfiability Testing ({SAT}~'10)}, volume 6175 of {\em
  Lecture Notes in Computer Science}, pages 388\nobreakdash--397. Springer,
  July 2010.

\end{thebibliography}

\bibliographystyle{alpha}

\end{document}